%% file: main.tex
\documentclass{article}
\usepackage[margin=1.1in]{geometry}

\PassOptionsToPackage{numbers,compress}{natbib}

\usepackage[ruled,vlined]{algorithm2e}
\usepackage[utf8]{inputenc} 
\usepackage[T1]{fontenc}    
\usepackage{hyperref}       
\usepackage{url}            
\usepackage{booktabs}       
\usepackage{amsfonts}       
\usepackage{nicefrac}       
\usepackage{microtype}      
\usepackage{xcolor}         
\usepackage{amsmath,amssymb,amsthm}
\usepackage{graphicx}
\usepackage[font=small,labelfont=bf]{caption}
\usepackage{natbib}
\usepackage{subcaption}
\usepackage{authblk}
\usepackage{float}
\usepackage{placeins}
\usepackage{tikz}
\usetikzlibrary{arrows.meta, positioning, fit, backgrounds, calc, matrix, decorations.pathreplacing}

\definecolor{blockblue}{RGB}{218,227,243}
\definecolor{blockgreen}{RGB}{222,234,211}
\definecolor{blockyellow}{RGB}{246,229,176}

\tikzset{
  flowarrow/.style={-Latex, line width=0.9pt},
  paneltitle/.style={font=\bfseries\large},
  explabel/.style={align=center, font=\bfseries},
  mathlabel/.style={align=center},
}

\input{macros}

\title{Sort, Partition, Randomize: Optimal Binary Hypothesis Testing under Local Differential Privacy}

\author[1]{Elena Ghazi\thanks{Corresponding author: \texttt{elenaghazi@g.harvard.edu}}}
\author[2]{Jawad Nasser}
\author[1]{Flavio Calmon}
\author[2]{Ibrahim Issa}
\affil[1]{Harvard University}
\affil[2]{American University of Beirut}

\date{}

\begin{document}

\maketitle

\begin{abstract}
We study optimal design of $\varepsilon$-locally differentially private mechanisms for binary hypothesis testing. Each observation is drawn from one of two known distributions $P_0,P_1$ on a finite alphabet of size $k$, privatized by a mechanism $Q$, and then used to infer which distribution generated the data. We measure testing utility using an $f$-divergence---including total variation, KL, and hockey-stick divergences---between the two induced output distributions. Previous work established structural properties of optimal mechanisms, but only yielded exponential-time algorithms. We prove a sharp structure: for every $\varepsilon$ and every $f$-divergence objective, after sorting the alphabet by likelihood ratio, there exists an optimal mechanism that partitions the sorted alphabet into contiguous blocks and applies randomized response to the block label. We call this class \emph{Sort-Partition-Randomize} (SPR). This characterization yields an exact dynamic program that computes an optimal mechanism in $O(k^3)$ time, and more generally in $O(\ell k^2)$ time with an $\ell$-output budget. Our results make it possible to efficiently compute and characterize the exact optimum across the full privacy range, beyond asymptotic privacy regimes.

\end{abstract}

\tableofcontents
\newpage

\input{sections/01_introduction}
\input{sections/02_setup_and_spr}
\input{sections/03_main_results}
\input{sections/04_geometry_of_spr_optimality}
\input{sections/05_dynamic_program}
\input{sections/06_conclusion}
\newpage
\bibliographystyle{unsrtnat}
\bibliography{references}
\newpage
\appendix
\input{appendix/A_extreme_refinements/A0_overview}
\input{appendix/B_convex_hull_spr}
\input{appendix/C_coarsening_spr}
\input{appendix/D_egamma_proofs}
\input{appendix/E_numerical_experiments}

\end{document}

%% file: macros.tex
\theoremstyle{plain}
\newtheorem{theorem}{Theorem}[section]
\newtheorem{lemma}[theorem]{Lemma}
\newtheorem{remark}[theorem]{Remark}
\newtheorem{corollary}[theorem]{Corollary}
\newtheorem{proposition}[theorem]{Proposition}
\newtheorem{definition}[theorem]{Definition}
\theoremstyle{definition}
\newtheorem{example}[theorem]{Example}
\theoremstyle{plain}

\newcommand{\eps}{\varepsilon}
\DeclareMathOperator*{\argmax}{arg\,max}
\hypersetup{
    colorlinks=true,
    linkcolor=black,
    filecolor=blue,
    urlcolor=blue,
    citecolor=black
}

\newcommand{\X}{\mathcal{X}}
\newcommand{\Y}{\mathcal{Y}}

%% file: sections/01_introduction.tex
\section{Introduction}
\label{sec:introduction}

Balancing privacy and utility is a central challenge in data disclosure control.
A canonical instance of this challenge is privacy-preserving binary hypothesis testing: a set of samples is drawn from one of two distributions, denoted by $P_0$ and $P_1$, and each sample is randomized to ensure privacy.
An analyst then performs a hypothesis test to infer which distribution generated the samples given the privatized outputs.
A natural measure of how much testing power survives randomization is an $f$-divergence between the privatized output distributions induced by $P_0$ and $P_1$.
Different choices of $f$ capture different operational quantities: total variation determines the probability of correct guessing under a balanced prior (i.e., when the two hypotheses are equally likely), $E_\gamma$ (or hockey-stick) divergences trace operating points along the Neyman--Pearson ROC curve, and KL divergence controls the exponential decay of error probabilities as the analyst accumulates samples~\cite{Duchi2013MinimaxRates,Asoodeh2021EGammaContraction,Pensia2025BinaryTesting}.

We study privacy-preserving binary hypothesis testing under local differential privacy (LDP).
LDP~\cite{Kasiviswanathan2011WhatCanWeLearnPrivately,Warner1965RandomizedResponse} is among the most stringent notions of privacy, requiring each sample to be randomized individually prior to disclosure.
However, this strictness comes at a substantial utility cost in statistical applications~\cite{Duchi2020TheRightComplexity,Asoodeh2024ContractionLDPMechanisms}.
A natural design objective is therefore to find LDP mechanisms that maximize an $f$-divergence between the privatized output distributions induced by $P_0$ and $P_1$.

Formally, for $\nu \in \{0,1\}$, an observation $X \sim P_\nu$ on a finite alphabet $\X$ of size $k$ is privatized through an $\eps$-locally differentially private channel $Q(\cdot \mid x)$ with finite output alphabet $\Y$, producing an output $Y \in \Y$. The $\eps$-LDP constraint requires that for all $x,x' \in \X$ and all $y \in \Y$,
\begin{align}
    Q(y \mid x) \le e^\eps Q(y \mid x').
\end{align}
Denote by $M_\nu := Q^\top P_\nu$ the output distribution induced by $P_\nu$ and $Q$. Given a convex function $f : \mathbb{R}_+ \to \mathbb{R}$ with $f(1) = 0$, our goal is to design $Q$ to maximize the $f$-divergence defined as:
\begin{align}
    D_f(M_0 \| M_1) := \sum_{y \in \Y} M_1(y) f\!\Big(\tfrac{M_0(y)}{M_1(y)}\Big).
\end{align}
Let $\mathcal{Q}_{\eps}$ be the set of $\eps$-LDP channels (with input alphabet $\X$). We can formulate our optimization as:
\begin{align}
    \label{eq:core-formulation}
    \max_{Q\in \mathcal Q_{\varepsilon}} D_f(Q^\top P_0\|Q^\top P_1).
\end{align}

We also consider the case in which the output alphabet is constrained to be less than or equal to a given size $\ell$ (referred to as a communication constraint by~\citet{Pensia2025BinaryTesting}). To wit, we replace $\mathcal{Q}_\eps$ in~\eqref{eq:core-formulation} by $\mathcal{Q}_{\eps,\ell}$, the set of $\varepsilon$-LDP channels with at most $\ell$ outputs:
\begin{align}
    \label{eq:core-formulation-output}
    \max_{Q\in \mathcal Q_{\varepsilon,\ell}} D_f(Q^\top P_0\|Q^\top P_1).
\end{align}

Despite the relative simplicity of the formulations, exact optimal mechanisms were previously understood only in special regimes: a binary-output mechanism for all $f$-divergences at sufficiently small $\varepsilon$, and randomized response on the original alphabet for KL at sufficiently large $\varepsilon$~\cite{Kairouz2016ExtremalMechanisms} (the latter requires $\ell \geq k$ in the constrained case). The intermediate regime, where privacy and utility must genuinely be balanced---and arguably the case of greatest practical interest---has remained computationally out of reach: although significant results on the structure of the optimizers have been derived by Kairouz et al.~\cite{Kairouz2016ExtremalMechanisms} and Pensia et al.~\cite{Pensia2025BinaryTesting}, the resulting algorithms are exponential in $k$ (in the unconstrained output size case) or in $\ell$ (in the constrained case).

\paragraph{Contributions.}

Our contributions in this paper are threefold:
\begin{itemize}
    \item[1)] We provide the \emph{first polynomial time algorithm} to find an optimal mechanism. In particular, we provide a dynamic program that runs in $O(k^3)$ in the unconstrained case, and $O(\ell k^2)$ in the constrained case.
\end{itemize}
Consequently, our results enable the design of optimal mechanisms for alphabet sizes that were previously computationally infeasible. We illustrate this in Figure~\ref{fig:kl-eps-two-panel} for the unconstrained optimization with $k = 100$.

\input{figures/figure_KL_comparison}

\begin{itemize}
    \item[2)] Our dynamic program follows from a sharp structural characterization of the optimizers of~\eqref{eq:core-formulation} and~\eqref{eq:core-formulation-output}, formally stated in Theorem~\ref{thm:extreme-points-spr}. Specifically, after sorting the input symbols by the likelihood ratios $P_0(x)/P_1(x)$, there is an optimal channel that partitions the sorted alphabet into contiguous blocks and applies randomized response to the block label (we illustrate an example in Figure~\ref{fig:llr-pattern-mapping}). We refer to mechanisms of this form as \emph{sort--partition--randomize (SPR) mechanisms.}
\end{itemize}
Beyond computational efficiency, SPR mechanisms yield a simple interpretable description of the optimal privatization strategy: aggregate symbols with similar likelihood ratios, then privatize only the resulting coarse label. Furthermore, the optimality of SPR mechanisms immediately implies that an optimal mechanism does not require more than $k$ outputs (as a partition of an alphabet of size $k$ cannot have more than $k$ blocks). This recovers a result by~\citet{Kairouz2016ExtremalMechanisms}. Moreover, it enables us to jointly address both formulations~\eqref{eq:core-formulation} and~\eqref{eq:core-formulation-output}. In fact, our approach shows that SPR mechanisms are optimal for any convex objective. However, the dynamic program further utilizes the decomposability of $f$-divergences.

\begin{itemize}
    \item[3)] We recover and extend optimality results when $D_f$ is the $E_\gamma$ divergence. In particular, there exists an optimal \emph{binary-output} mechanism, for every $\ell \ge 2$, which can be viewed as a privatized version of the Neyman-Pearson threshold.
\end{itemize}
\citet{Zamanlooy2024E_gammaMixing} provided an upper bound on the optimal value in~\eqref{eq:core-formulation} for the $E_\gamma$ divergence. The upper bound was known to be tight for binary input distributions. We show that this result follows from our structural characterization theorem, and prove that it is tight for any input alphabet size.

\paragraph{Prior work.}
\citet{Kairouz2016ExtremalMechanisms} showed that, for a broad class of utilities including $f$-divergences, an optimal $\eps$-LDP mechanism may be chosen to be \emph{staircase}: for each output $y$, the probabilities $Q(y \mid x)$ take only two values, $\theta_y$ and $e^\eps \theta_y$, as $x$ varies. This reduces the optimization over arbitrary mechanisms to a finite linear program with one variable for each possible high set $H \subseteq \X$, hence $2^k$ variables. They also identified simple optimal mechanisms in extreme regimes: a binary-output mechanism for sufficiently small $\eps$ for any $f$-divergence, and randomized response on the original alphabet (i.e., $k$-ary randomized response) for sufficiently large $\eps$ in the KL case, with distribution-dependent $\eps$ thresholds. Our work replaces this exponential search by an $O(k^3)$ dynamic program, giving exact optimization across the full privacy range.

\citet{Pensia2025BinaryTesting} studied the geometry of achievable pairs of output distributions under LDP and communication constraints. In the $\eps$-LDP setting with an $\ell$-output budget, they showed that extreme achievable pairs can be realized by first quantizing the input using likelihood-ratio thresholds into at most $2\ell^2$ intermediate symbols, and then applying an extreme $\eps$-LDP channel from this intermediate alphabet to the $\ell$ outputs. This gives an algorithm with runtime polynomial in $k^{\ell^2}$ and $2^{O(\ell^3\log \ell)}$~\cite[Corollary 4]{Pensia2025BinaryTesting}, hence polynomial in $k$ only when $\ell$ is treated as a constant. Building on this geometric perspective and a perturbation-based argument, we sharpen the structure substantially: in the pure LDP, $f$-divergence setting, the intermediate quantizer and arbitrary private channel collapse to an SPR mechanism, yielding the $O(\ell k^2)$ dynamic program.

\citet{Tsitsiklis1993ExtremalProperties} studied the non-private analogue and showed the optimality of likelihood-ratio quantizers for broad classes of binary testing and quantization problems. Our SPR mechanisms are the private counterpart: sort by likelihood ratio, partition into contiguous blocks, and then apply randomized response. As $\eps \to \infty$, randomized response becomes deterministic, so SPR mechanisms recover the non-private likelihood-ratio quantizers.

The remainder of this paper is organized as follows.
Section~\ref{sec:setup-and-spr} fixes notation and introduces staircase and SPR mechanisms.
Section~\ref{sec:main-results-hub} states our main theorems: SPR optimality, the dynamic program, the closed-form mechanism for $E_\gamma$-divergences, and consequences for R\'enyi divergences and $\ell_r$ distance objectives.
Section~\ref{sec:geometry-of-spr} develops the joint-range geometry behind SPR optimality and proves Theorem~\ref{thm:extreme-points-spr} and Corollary~\ref{cor:fdiv-spr-optimality}.
Section~\ref{sec:dynamic-program} gives the $O(\ell k^2)$ dynamic program, including its runtime, traceback, and pseudocode.

%% file: figures/figure_KL_comparison.tex
\begin{figure}[tbp]
  \centering
  \begin{subfigure}{0.49\linewidth}
    \centering
    \includegraphics[width=\linewidth]{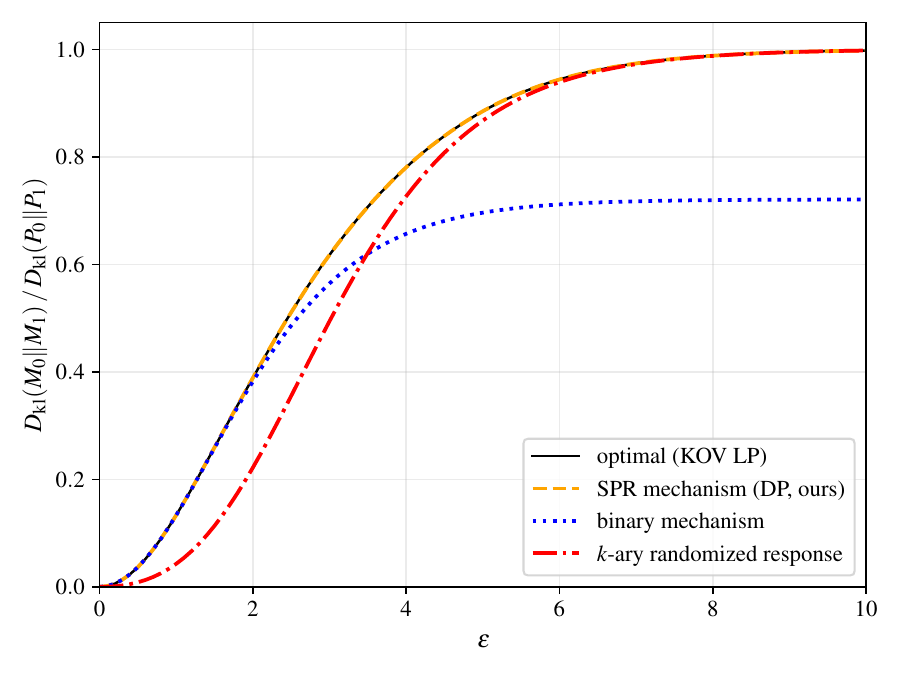}
    \caption{$k=6$.}
    \label{fig:kl-eps-k6}
  \end{subfigure}\hfill
  \begin{subfigure}{0.49\linewidth}
    \centering
    \includegraphics[width=\linewidth]{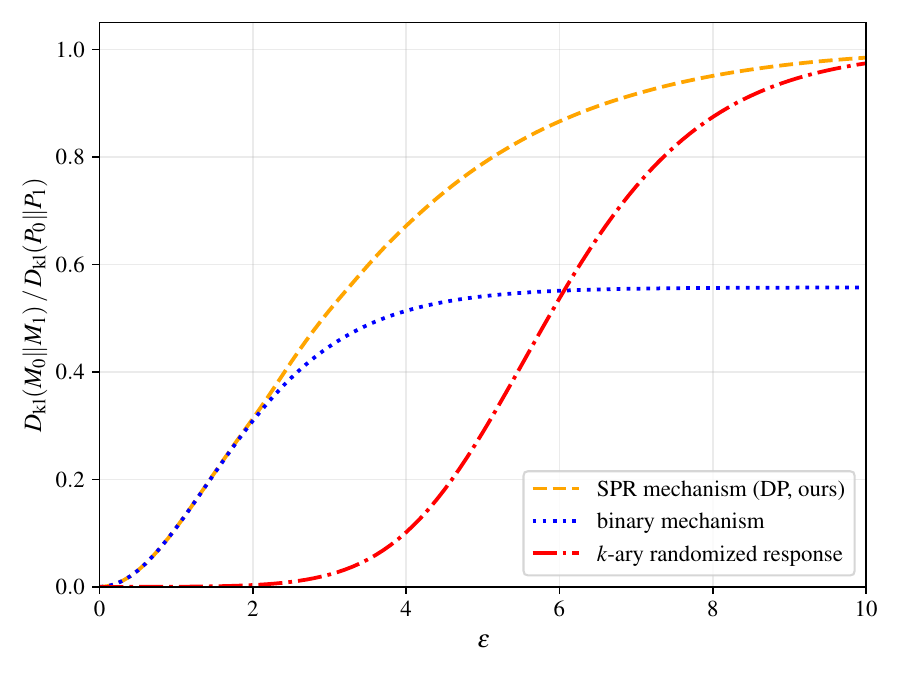}
    \caption{$k=100$.}
    \label{fig:kl-eps-k100}
  \end{subfigure}
  \caption{Average normalized KL utility $D_{\mathrm{kl}}(M_0\|M_1)/D_{\mathrm{kl}}(P_0\|P_1)$ versus $\eps$, averaged over $T=100$ Dirichlet$(\mathbf 1_k)$ pairs $(P_0,P_1)$ (NumPy seed $0$, $\eps$ grid $\{0,0.1,\dots,10\}$). \subref{fig:kl-eps-k6} $k=6$, where the KOV LP is tractable. \subref{fig:kl-eps-k100} $k=100$, where the KOV LP is infeasible. Full setup, mechanism implementations, runtimes, and bootstrap bands are in Appendix~\ref{sec:experiments}.}
  \label{fig:kl-eps-two-panel}
\end{figure}

%% file: sections/02_setup_and_spr.tex
\section{Problem setup and SPR mechanisms}
\label{sec:setup-and-spr}

This section fixes notation and introduces the two structural objects on which the rest of the paper is built: staircase mechanisms and sort--partition--randomize (SPR) mechanisms.

\paragraph{Notation.}
We recall the setup from Section~\ref{sec:introduction}. An observation $X\sim P_\nu$ for $\nu\in\{0,1\}$ on a finite alphabet $\X$ is privatized through an $\eps$-LDP channel $Q$ with finite output alphabet $\Y$, i.e., $Q(y\mid x)\le e^\eps Q(y\mid x')$ for all $x,x'\in\X$ and $y\in\Y$, producing an output $Y\in\Y$. The induced output marginals are $M_\nu := Q^\top P_\nu$, and utility is measured by an $f$-divergence $D_f(M_0\|M_1)$ for a convex $f:\mathbb R_+\to\mathbb R$ with $f(1)=0$. The goal is to design $Q$ to maximize $D_f(M_0\|M_1)$.

\paragraph{Likelihood-ratio ordering and reduction of ties.}
We first remove a few degenerate cases so that the likelihood-ratio order is strict, which is the setting used in intermediate arguments.
Deleting symbols with zero mass under both hypotheses and merging symbols with the same likelihood ratio $P_0(x)/P_1(x)$ does not change the induced output pair $(M_0,M_1)$ of any $\eps$-LDP mechanism: within each likelihood-ratio class, replacing the corresponding rows by the appropriate weighted average preserves $\eps$-LDP and leaves both induced marginals unchanged. Conversely, any mechanism on the reduced alphabet can be lifted to the original alphabet by copying its row to all symbols in the corresponding class. We relabel the reduced likelihood-ratio-sorted alphabet as $[k]$, and write $p,q$ for the reduced distributions. Thus $k$ may be smaller than $|\X|$, and
\begin{align}
    0 \le r_1 < r_2 < \cdots < r_k \le \infty,
    \qquad
    r_i := \frac{p_i}{q_i}.
\end{align}

\paragraph{Output constraints and the joint range.}
Let $\mathcal Q_\eps$ denote the set of $\eps$-LDP channels with input alphabet $[k]$. For $\ell\ge 1$, let $\mathcal Q_{\eps,\ell}$ denote the set of $\ell$-output $\eps$-LDP channels,
\begin{align}
    \mathcal Q_{\eps,\ell}
    :=
    \Bigl\{
    Q\in\mathbb R_+^{k\times \ell}:
    Q\mathbf{1}_\ell=\mathbf{1}_k,\ Q \text{ is }\eps\text{-LDP}
    \Bigr\},
\end{align}
and the fixed-output joint range
\begin{align}
    \mathcal R_\ell(p,q)
    :=
    \Bigl\{
    (Q^\top p,Q^\top q): Q\in\mathcal Q_{\eps,\ell}
    \Bigr\}.
\end{align}

We assume $\varepsilon>0$ throughout. When $\varepsilon=0$, every feasible channel is input-independent, so $Q^\top p=Q^\top q$ and every $f$-divergence objective is zero.

\paragraph{Staircase mechanisms.}

For any channel $Q$, an output $y\in\Y$ is \emph{active} if $Q(y\mid\cdot)$ is not identically zero.

\begin{definition}[Staircase mechanism and high set]
    \label{def:staircase}
    An $\eps$-LDP mechanism $Q$ is a \textbf{staircase mechanism} if, for every active $y$ and for all $x$ and $x'$, $Q(y|x)/Q(y|x') \in \{e^{-\eps},1,e^\eps\}.$ Moreover, for each $y$, we define the \textbf{high set} $H_y \subseteq \X$ as:
    \begin{align}
        H_y = \{x \in \X: Q(y|x) = \max_{x' \in \X} Q(y|x') \}.
    \end{align}
\end{definition}

If $Q$ is staircase, it will be convenient to express each active column $y$ as
\begin{align}
    Q(i,y)=\theta_y\Bigl(1+(e^\eps-1)\mathbf{1}\{i\in H_y\}\Bigr),
    \qquad \theta_y>0,\ H_y\subseteq [k].
\end{align}
We call $\theta_y$ the column scale and $H_y$ the \emph{high set} of column $y$. When an active staircase column is constant, i.e., $Q(i,y)\equiv \alpha_y>0$,  we write it as
\begin{align}
    Q(i,y)=\alpha_y e^{-\eps}\Bigl(1+(e^\eps-1)\mathbf{1}\{i\in [k]\}\Bigr).
\end{align}
With this convention, every active high set is nonempty.

\paragraph{SPR mechanisms and example.}
If $\tau:[m]\to[\ell]$ is a deterministic map and $Q\in\mathbb R_+^{k\times m}$, we write $\tau\circ Q\in\mathbb R_+^{k\times \ell}$ for the post-processed channel $(\tau\circ Q)(i,y):=\sum_{z:\tau(z)=y} Q(i,z)$. An \emph{LR-contiguous partition} of $[k]$ is a partition $\pi=(B_1,\dots,B_s)$ into nonempty intervals in the likelihood-ratio order.

\begin{definition}[Sort--partition--randomize mechanism]
    \label{def:spr}
    A \emph{sort--partition--randomize} (SPR) mechanism first sorts the input symbols as $x_1,\dots,x_k$ by likelihood ratio $P_0(x_i)/P_1(x_i)$, partitions this ordered alphabet into contiguous blocks $\pi=(B_1,\dots,B_s)$, and then applies $s$-ary randomized response (RR) to the block label:
    \begin{align}
        Q^\pi(j \mid x)=
        \begin{cases}
            \dfrac{e^\eps}{e^\eps+s-1}, & x\in B_j,\\[5pt]
            \dfrac{1}{e^\eps+s-1}, & x\notin B_j .
        \end{cases}
    \end{align}
    Equivalently, $Q^\pi$ is a staircase mechanism whose active high sets are exactly the likelihood-ratio-contiguous blocks $B_1,\dots,B_s$. When needed, we may also view $Q^\pi$ as a channel with a larger output alphabet by relabeling its $s$ active outputs and padding the remaining columns with zeros.
\end{definition}

\begin{example}
Figure~\ref{fig:llr-pattern-mapping} illustrates an SPR mechanism. Here $k=6$, the inputs are already ordered by likelihood ratio, and the partition has $s=3$ contiguous blocks. Each row has one high entry and two low entries, so after normalization the channel is exactly $3$-ary RR on the block label.
\end{example}
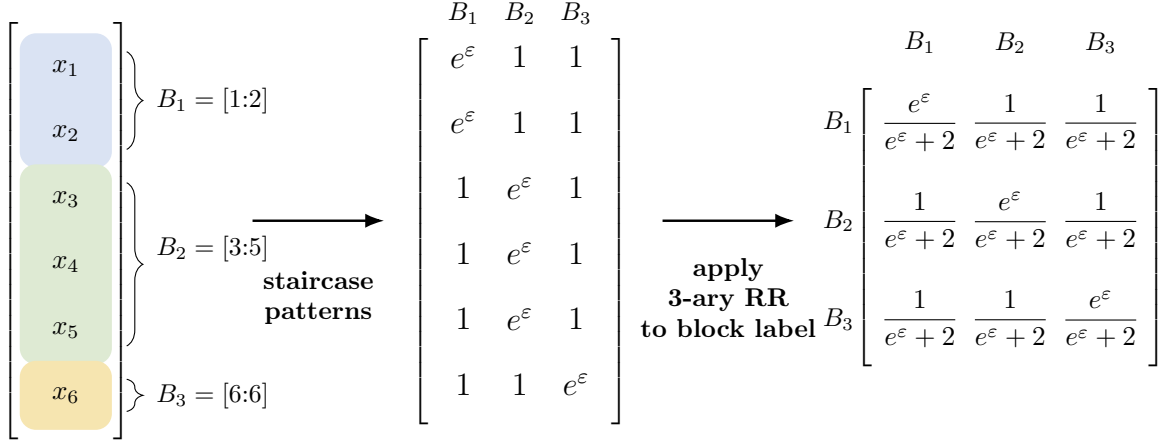
\begin{figure}[tbp]
  \centering
  \resizebox{0.97\linewidth}{!}{\input{figures/illustrative_example}}
  \caption{An LR-contiguous partition and the corresponding staircase pattern, with inputs labeled so that $r_1 \le r_2 \le \cdots \le r_6$ where $r_i = p_i/q_i$. After normalization, this becomes $3$-ary RR on the block label.}
  \label{fig:llr-pattern-mapping}
\end{figure}

\begin{remark}
The binary mechanism of \citet{Kairouz2016ExtremalMechanisms} is a special case with $s=2$. Their KL-optimal randomized response mechanism in the sufficiently low-privacy / large-$\eps$ regime is the special case in which every block is a singleton ($k$-ary RR).
\end{remark}

%% file: figures/illustrative_example.tex
\begin{tikzpicture}[scale=0.95, transform shape]

    \matrix (xlist) [
      matrix of math nodes,
      left delimiter={[},
      right delimiter={]},
      nodes={minimum width=9mm, minimum height=6mm, anchor=center},
      row sep=2.5mm
    ] at (3.5,1.2) {
      |[minimum height=12pt, inner sep=0pt]| {}\\[-2.5mm]
      x_1\\
      x_2\\
      x_3\\
      x_4\\
      x_5\\
      x_6\\[-2.5mm]
      |[minimum height=4pt, inner sep=0pt]| {}\\
    };

    \begin{scope}[on background layer]
      \node[fit=(xlist-2-1)(xlist-3-1), fill=blockblue, rounded corners=6pt, inner sep=4pt] {};
      \node[fit=(xlist-4-1)(xlist-6-1), fill=blockgreen, rounded corners=6pt, inner sep=4pt] {};
      \node[fit=(xlist-7-1)(xlist-7-1), fill=blockyellow, rounded corners=6pt, inner sep=4pt] {};
    \end{scope}

    \draw[decorate, decoration={brace, amplitude=6pt}]
      ($(xlist-2-1.east)+(0.35,0.22)$) -- ($(xlist-3-1.east)+(0.35,-0.22)$)
      node[midway, right=8pt] {$B_1=[1{:}2]$};

    \draw[decorate, decoration={brace, amplitude=6pt}]
      ($(xlist-4-1.east)+(0.35,0.22)$) -- ($(xlist-6-1.east)+(0.35,-0.22)$)
      node[midway, right=8pt, yshift=6pt] {$B_2=[3{:}5]$};

    \draw[decorate, decoration={brace, amplitude=6pt}]
      ($(xlist-7-1.east)+(0.35,0.22)$) -- ($(xlist-7-1.east)+(0.35,-0.22)$)
      node[midway, right=8pt] {$B_3=[6{:}6]$};

    \draw[flowarrow] (6.05,1.2) -- (7.85,1.2);
    \node[explabel] at (6.95,0.15)
      {staircase\\ patterns};

    \matrix (smat) [
      matrix of math nodes,
      left delimiter={[},
      right delimiter={]},
      nodes={minimum width=4.5mm, minimum height=6mm, anchor=center, font=\large},
      column sep=1.5mm,
      row sep=2.5mm
    ] at (9.7,1.2) {
      |[minimum height=4pt, inner sep=0pt]| {} & |[minimum height=4pt, inner sep=0pt]| {} & |[minimum height=4pt, inner sep=0pt]| {}\\[-2.5mm]
      e^{\varepsilon} & 1 & 1\\
      e^{\varepsilon} & 1 & 1\\
      1 & e^{\varepsilon} & 1\\
      1 & e^{\varepsilon} & 1\\
      1 & e^{\varepsilon} & 1\\
      1 & 1 & e^{\varepsilon}\\[-2.5mm]
      |[minimum height=4pt, inner sep=0pt]| {} & |[minimum height=4pt, inner sep=0pt]| {} & |[minimum height=4pt, inner sep=0pt]| {}\\
    };

    \node[above=0pt of smat-2-1] {$B_1$};
    \node[above=0pt of smat-2-2] {$B_2$};
    \node[above=0pt of smat-2-3] {$B_3$};

    \draw[flowarrow] (11.65,1.2) -- (13.45,1.2);
    \node[explabel] at (12.55,0.15)
      {apply\\ 3-ary RR\\ to block label};

    \matrix (rr) [
      matrix of math nodes,
      left delimiter={[},
      right delimiter={]},
      nodes={minimum width=9mm, minimum height=8mm, anchor=center, inner xsep=1pt},
      column sep=1mm,
      row sep=3mm,
      inner xsep=1pt
    ] at (16.4,1.2) {
      \dfrac{e^\varepsilon}{e^\varepsilon+2} & \dfrac{1}{e^\varepsilon+2} & \dfrac{1}{e^\varepsilon+2}\\
      \dfrac{1}{e^\varepsilon+2} & \dfrac{e^\varepsilon}{e^\varepsilon+2} & \dfrac{1}{e^\varepsilon+2}\\
      \dfrac{1}{e^\varepsilon+2} & \dfrac{1}{e^\varepsilon+2} & \dfrac{e^\varepsilon}{e^\varepsilon+2}\\
    };

    \node[above=3mm of rr-1-1] {$B_1$};
    \node[above=3mm of rr-1-2] {$B_2$};
    \node[above=3mm of rr-1-3] {$B_3$};

    \node[left=2mm of rr-1-1] {$B_1$};
    \node[left=2mm of rr-2-1] {$B_2$};
    \node[left=2mm of rr-3-1] {$B_3$};

\end{tikzpicture}

%% file: sections/03_main_results.tex
\section{Main results}
\label{sec:main-results-hub}

This section collects the headline statements of the paper. Section~\ref{sec:main-results-spr} states the structural characterization of optimal mechanisms; Section~\ref{sec:main-results-dp} states the resulting polynomial-time algorithm; Sections~\ref{sec:main-results-egamma} and \ref{sec:main-results-other} state consequences for $E_\gamma$-divergences and other utility objectives.

\subsection{SPR optimality theorem}
\label{sec:main-results-spr}

Our main structural result is that, after sorting the alphabet by likelihood ratio, every extreme point of the fixed-output joint range $\mathcal R_\ell(p,q)$ is attained by an SPR mechanism.

\begin{theorem}[Extreme points are sort--partition--randomize]
\label{thm:extreme-points-spr}
If $(u,v)\in\mathcal R_\ell(p,q)$ is an extreme point, then there exists an LR-contiguous partition $\pi=(B_1,\dots,B_s)$ of $[k]$ such that $(u,v)$ is induced by an $\ell$-output channel obtained from $Q^\pi$ by relabeling its active outputs and padding with zero columns. In particular, necessarily $s\le \ell$ and $s\le k$. Equivalently, every extreme point of $\mathcal R_\ell(p,q)$ is realized by a sort--partition--randomize mechanism.
\end{theorem}

Since every $f$-divergence is convex in the induced pair $(M_0,M_1)$, the maximum over $\mathcal R_\ell(p,q)$ is attained at an extreme point. Theorem~\ref{thm:extreme-points-spr} therefore yields the following optimality statement.

\begin{corollary}[SPR optimality for $f$-divergences]
\label{cor:fdiv-spr-optimality}
For every $f$-divergence, the following hold.
\begin{enumerate}
    \item Over all finite-output $\eps$-LDP mechanisms, the optimum is attained by an SPR mechanism with at most $k$ active outputs.
    \item For any output budget $1\le \ell\le k$, $\max_{Q\in\mathcal Q_{\eps,\ell}} D_f(Q^\top p\|Q^\top q)$ is attained by an SPR mechanism with at most $\ell$ active outputs.
\end{enumerate}
\end{corollary}

\subsection{Exact dynamic program}
\label{sec:main-results-dp}

Corollary~\ref{cor:fdiv-spr-optimality} reduces the search for an optimal mechanism to a search over LR-contiguous partitions of $[k]$. We exploit the additive decomposability of $f$-divergences across blocks to obtain an exact dynamic program.

For $z\in\mathbb{R}_+^k$, let $\displaystyle \mu(z):=(q^\top z)\,f\!\left(\frac{p^\top z}{q^\top z}\right)$ be the column score associated with the chosen $f$-divergence. For $1\le a\le b\le k$, let
\begin{align}
    \mu[a:b]
    :=
    \mu\!\left(
        \bigl(1+(e^\varepsilon-1)\mathbf{1}\{a\le i\le b\}\bigr)_{i=1}^k
    \right).
\end{align}

\begin{proposition}[Dynamic program with an output budget]
\label{prop:dp-contiguous-partition}
Fix $1\le \ell\le k$. For $1\le s\le \ell$ and $s\le i\le k$, let $F[s,i]$ be the maximum raw score over all partitions of $[1:i]$ into exactly $s$ nonempty contiguous blocks. Then $F[1,i]=\mu[1:i]$, and for $s\ge 2$,
\begin{align}
    F[s,i]=\max_{s-1\le t<i}\Bigl\{F[s-1,t]+\mu[t+1:i]\Bigr\}.
\end{align}
The optimal $\ell$-output utility is $\displaystyle \max_{1\le s\le \ell}\frac{F[s,k]}{e^\varepsilon+s-1}$.
\end{proposition}

\begin{corollary}[Runtime]
\label{cor:dp-runtime}
For an output budget $1\le \ell\le k$, the optimal value and an optimal SPR mechanism can be computed in $O(\ell k^2)$ time and $O(\ell k)$ space. In particular, the unconstrained-output problem is obtained by taking $\ell=k$, giving $O(k^3)$ time and $O(k^2)$ space.
\end{corollary}

The proofs of Proposition~\ref{prop:dp-contiguous-partition} and Corollary~\ref{cor:dp-runtime}, together with pseudocode for the algorithm, are given in Section~\ref{sec:dynamic-program}.

\subsection{Closed-form optimal mechanism for \texorpdfstring{$E_\gamma$-divergences}{E-gamma-divergences}}
\label{sec:main-results-egamma}

For $\gamma\ge 1$, define the $E_\gamma$- (or hockey-stick) divergence
\begin{align}
    E_\gamma(P_0 \| P_1) := \sup_{A\subseteq\X} \bigl( P_0(A) - \gamma P_1(A) \bigr),
\end{align}
which corresponds to the convex function $f_\gamma(t) = (t-\gamma)_+$, with $\gamma=1$ recovering total variation. The $E_\gamma$-divergence has a direct testing interpretation. If a test decides $P_0$ on an event $A\subseteq\Y$, then $M_0(A^c)$ is its missed-detection probability and $M_1(A)$ is its false-alarm probability. Therefore
\begin{align}
    E_\gamma(M_0\|M_1)
    =
    1-\inf_{A\subseteq\Y}\{M_0(A^c)+\gamma M_1(A)\}.
\end{align}
Maximizing $E_\gamma$ is equivalent to minimizing this weighted testing risk: $\gamma=1$ gives total variation, while $\gamma>1$ captures asymmetric costs. Such asymmetry is common in privacy-sensitive decisions (rare-event detection, medical screening, fraud or abuse detection, and content moderation) where false alarms and missed detections need not have comparable costs.

For the $E_\gamma$-divergences, the SPR dynamic program collapses to a closed form.

\begin{definition}[Generalized binary mechanism]
\label{def:generalized-binary-mechanism}
With $S_\gamma := \{ x \in \X : P_0(x) \ge \gamma P_1(x) \}$, the \emph{generalized binary mechanism} is the binary-output channel $Q_\gamma$ defined by
\begin{align}
    Q_\gamma(0 \mid x) &= \begin{cases}
        \dfrac{e^\eps}{1 + e^\eps}, & x \in S_\gamma, \\[1.6ex]
        \dfrac{1}{1 + e^\eps},      & x \notin S_\gamma,
    \end{cases}
    & Q_\gamma(1 \mid x) &= 1 - Q_\gamma(0 \mid x).
\end{align}
\end{definition}

\begin{theorem}[Optimal $E_\gamma$ mechanism]
\label{thm:optimal-egamma}
For every $\gamma \ge 1$ and every $\eps$-LDP mechanism $Q$ with induced marginals $M_0, M_1$,
\begin{align}
    \label{eq:upper_bound_e_gamma}
    E_\gamma(M_0 \| M_1)
    \le
    \left(
        \frac{e^\eps - 1}{e^\eps + 1} E_\gamma(P_0 \| P_1)
        + \frac{1 - \gamma}{e^\eps + 1}
    \right)_+ .
\end{align}
The generalized binary mechanism $Q_\gamma$ attains \eqref{eq:upper_bound_e_gamma} with equality, and therefore maximizes $E_\gamma(M_0 \| M_1)$ over $\mathcal Q_\eps$.
\end{theorem}

In the non-private problem, the event maximizing $P_0(A)-\gamma P_1(A)$ is the threshold set $S_\gamma=\{x\in\X:P_0(x)\ge\gamma P_1(x)\}$, obtained by selecting exactly the points where $P_0(x)-\gamma P_1(x)\ge 0$. Theorem~\ref{thm:optimal-egamma} states that privacy does not change this threshold: the optimal private mechanism applies binary randomized response to the bit $\mathbf 1\{X\in S_\gamma\}$. Notably, $S_\gamma$ does not depend on $\varepsilon$: privacy leaves the classical likelihood-ratio threshold unchanged and privatizes only the resulting bit. This captures a private Neyman--Pearson lemma for the weighted-risk formulation. Since hockey-stick divergences also underlie privacy profiles and worst-case LDP contraction bounds~\cite{Balle2020PrivacyProfiles,Asoodeh2021EGammaContraction}, this result identifies, for the fixed pair $(P_0,P_1)$, the $\eps$-LDP channel that preserves the largest possible $E_\gamma$-separation.

\citet{Zamanlooy2024E_gammaMixing} (Theorem~3) established the upper bound \eqref{eq:upper_bound_e_gamma}; Theorem~\ref{thm:optimal-egamma} shows that the generalized binary mechanism $Q_\gamma$ attains it. The argument is short: under $Q_\gamma$, the contribution of the output $0$ to $E_\gamma(M_0\|M_1)$ is $M_0(0)-\gamma M_1(0)=\sum_x(P_0(x)-\gamma P_1(x))Q_\gamma(0\mid x)$. Splitting this sum over $S_\gamma$ and $S_\gamma^c$ and using the Neyman--Pearson identity $E_\gamma(P_0\|P_1)=P_0(S_\gamma)-\gamma P_1(S_\gamma)$ yields exactly the right-hand side of \eqref{eq:upper_bound_e_gamma}. The full proof is given in Appendix~\ref{sec:proof-optimal-egamma}; the same appendix shows that the upper bound \eqref{eq:upper_bound_e_gamma} can also be recovered directly from SPR optimality, by analyzing the contribution of each block of an LR-contiguous partition.

\subsection{Other objectives: R\'enyi divergences and \texorpdfstring{$\ell_r$ distance}{l-r distance}}
\label{sec:main-results-other}

SPR optimality applies to any objective maximized at extreme points of the joint-range polytope; the $O(\ell k^2)$ dynamic program, however, requires an additive decomposition over contiguous blocks, as $f$-divergences have. R\'enyi divergences fit the same framework via a monotone reduction to an $f$-divergence; for $\ell_r$ distance, the optimum coincides with the TV-optimal binary SPR mechanism.

\begin{proposition}[SPR optimality for R\'enyi divergences]
\label{prop:renyi_divergence_optimality}
For every fixed finite order $\alpha>0$, the SPR optimality result and the $O(\ell k^2)$ dynamic program apply to maximization of $D_\alpha$.
\end{proposition}

\begin{proof}
The case $\alpha=1$ is KL divergence, an $f$-divergence. For $\alpha\neq 1$,
\begin{align}
    D_\alpha(u\|v) = \frac{1}{\alpha-1}\log\sum_y u_y^\alpha v_y^{1-\alpha},
\end{align}
so maximizing $D_\alpha$ is equivalent to maximizing the $f$-divergence with $f(t)=t^\alpha-1$ for $\alpha>1$ and $f(t)=1-t^\alpha$ for $0<\alpha<1$; both are convex with $f(1)=0$.
\end{proof}

\begin{proposition}[Optimal mechanism for $\ell_r$ distance is binary]
\label{prop:optimality_l_p}
For any $1\le r\le\infty$ and any $\eps$-LDP channel $Q$ with at least two outputs,
\begin{align}
    \|Q^\top p-Q^\top q\|_r \le 2^{1/r}\tanh(\eps/2)\,d_{\rm TV}(p,q), \qquad 2^{1/\infty}:=1.
\end{align}
The maximum is achieved by binary randomized response applied to $S:=\{i\in[k]:p_i\ge q_i\}$.
\end{proposition}

\begin{proof}
Let $\Delta := Q^\top p-Q^\top q$ and $m:=d_{\rm TV}(Q^\top p,Q^\top q)$. By the total-variation contraction for $\eps$-LDP channels,
\begin{align}
    m\le \tanh(\eps/2)\,d_{\rm TV}(p,q).
\end{align}
The positive entries of $\Delta$ sum to $m$, and the negative entries have absolute values summing to $m$. Hence, for $1\le r<\infty$,
\begin{align}
    \|\Delta\|_r^r\le m^r+m^r=2m^r,
\end{align}
therefore,
\begin{align}
    \|\Delta\|_r
    \le 2^{1/r}\tanh(\eps/2)\,d_{\rm TV}(p,q).
\end{align}
For $r=\infty$, the same argument gives
\begin{align}
    \|\Delta\|_\infty
    \le \tanh(\eps/2)\,d_{\rm TV}(p,q).
\end{align}
For tightness, let $S=\{i\in[k]:p_i\ge q_i\}$ and apply binary randomized response to $\mathbf 1\{i\in S\}$. Since $p(S)-q(S)=d_{\rm TV}(p,q)$, the output difference vector is $(a,-a)$ where $a:=\tanh(\eps/2)\,d_{\rm TV}(p,q)$. Its $\ell_r$ norm is $2^{1/r}a$, with $2^{1/\infty}:=1$.
\end{proof}

\subsection{Numerical illustration}
\label{sec:experiments-main}

Figure~\ref{fig:kl-eps-two-panel} illustrates SPR optimality on the KL utility for $\eps$-LDP mechanisms. We compare the SPR dynamic program (Algorithm~\ref{alg:optimal-spr-dp}) against the closed-form binary mechanism, $k$-ary randomized response, and (where tractable) the Kairouz--Oh--Viswanath staircase linear program~\cite{Kairouz2016ExtremalMechanisms}, on $T=100$ Dirichlet$(\mathbf 1_k)$ pairs $(P_0,P_1)$ at each $\eps$ on a uniform grid in $[0,10]$. For $k=6$ the LP is tractable and the SPR DP curve coincides with it to numerical precision ($\le 6.1\times 10^{-14}$ across the $T\times 101$ grid), validating Theorem~\ref{thm:extreme-points-spr} empirically. For $k=100$ the LP has $2^{100}\approx 1.27\times 10^{30}$ variables and is omitted, while the SPR DP completes the entire sweep in under $20$ seconds on a laptop. The binary mechanism plateaus and $k$-ary RR is uninformative for small $\eps$; the SPR optimum interpolates between them across the full privacy range. Full setup, bootstrap confidence bands, and runtime details are reported in Appendix~\ref{sec:experiments}.

%% file: sections/04_geometry_of_spr_optimality.tex
\section{Geometry of SPR optimality}
\label{sec:geometry-of-spr}

This section assembles the geometric ingredients behind the SPR optimality theorem (Theorem~\ref{thm:extreme-points-spr}) and its $f$-divergence corollary (Corollary~\ref{cor:fdiv-spr-optimality}). The argument has three steps, each carried out by one proposition: every extreme image point can be realized by a staircase channel whose active columns are LR-contiguous (\S\ref{sec:geom-extreme-refinement}); every such staircase channel lies in the convex hull of SPR mechanisms (\S\ref{sec:geom-convex-hull}); and deterministically merging active outputs of an SPR mechanism never produces a new extreme point (\S\ref{sec:geom-coarsening}). Section~\ref{sec:geom-proof} chains these three propositions to prove Theorem~\ref{thm:extreme-points-spr} and Corollary~\ref{cor:fdiv-spr-optimality}. Figure~\ref{fig:joint-range-k5-eps2} previews the resulting polytope for a representative $(p,q)$ at $k=5$, $\ell=2$, $\varepsilon=2$.

\begin{figure}[tbp]
  \centering
  \includegraphics[width=0.72\linewidth]{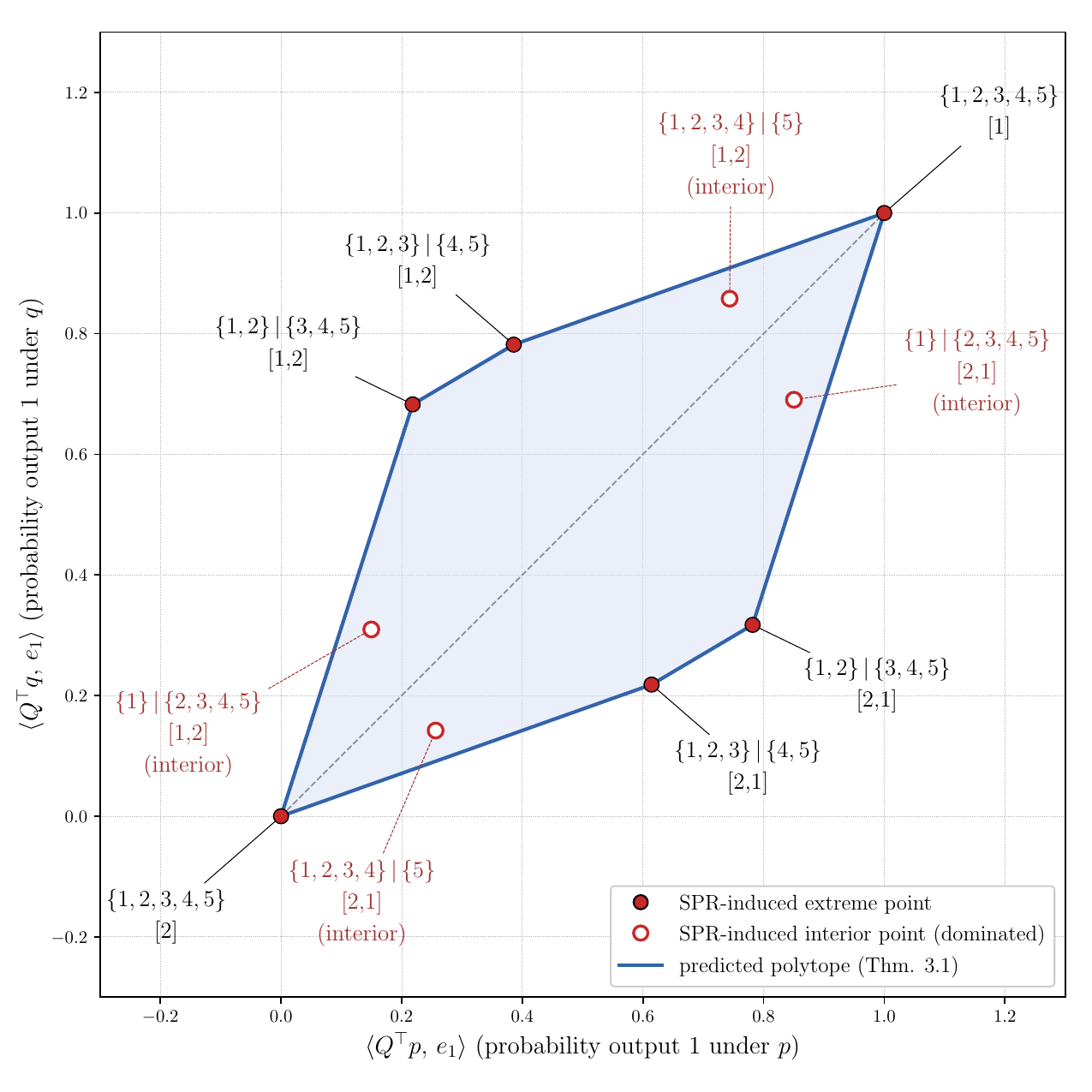}
  \caption{Geometry of SPR optimality for $k=5$, $\ell=2$, $\varepsilon=2$, with $p=(0.04, 0.09, 0.22, 0.47, 0.18)$ and $q=(0.25, 0.49, 0.13, 0.10, 0.03)$. The joint range $\mathcal{R}_2(p,q) = \{(Q^\top p, Q^\top q) : Q\in\mathcal{Q}_{\varepsilon,2}\}$ is the shaded hexagon, equal to the convex hull of the $2(k-1)+2=10$ SPR-induced candidate points (Theorem~\ref{thm:extreme-points-spr}). Since $\ell=2$, the plotted coordinates determine the full pair $(Q^\top p,Q^\top q)$: the omitted coordinates are their complements. Thus this two-dimensional plot is a lossless representation of $\mathcal R_2(p,q)$. Each candidate is the image of one LR-contiguous partition into two blocks together with a block-to-output assignment $[a,b]$ meaning $B_1\to a,\ B_2\to b$, plus the two constant $s=1$ mechanisms at $(0,0)$ and $(1,1)$. Six of the ten are \emph{extreme} (filled red); the other four are \emph{dominated} (hollow red): they are SPR mechanisms but sit strictly inside the polytope as convex combinations of others, so by Corollary~\ref{cor:fdiv-spr-optimality} they cannot maximize any $f$-divergence. Whether a given SPR mechanism is extreme depends on $(p,q)$ and $\varepsilon$.}
  \label{fig:joint-range-k5-eps2}
\end{figure}

\subsection{Extreme refinements: from extreme points to LR-contiguous staircase channels}
\label{sec:geom-extreme-refinement}

We first show that every extreme image point can be realized, up to deterministic post-processing, by a staircase channel whose active columns respect the likelihood-ratio order.

\begin{proposition}[Extreme points admit LR-contiguous staircase refinements]
\label{prop:extreme-refinement}
If $(u,v)\in\mathcal R_\ell(p,q)$ is an extreme point, then there exist an integer $m\ge 1$, a staircase channel $\widetilde Q\in\mathcal Q_{\eps,m}$, and a deterministic map $\tau:[m]\to[\ell]$ such that
\begin{align}
    (u,v)=\bigl((\tau\circ \widetilde Q)^\top p,(\tau\circ \widetilde Q)^\top q\bigr),
\end{align}
and every active column of $\widetilde Q$ has an LR-contiguous high set.
\end{proposition}

The proof is given in Appendix~\ref{sec:proof-extreme-refinement}. Starting from an arbitrary extreme image point, we refine it to a staircase channel and then show that any noncontiguous active column creates a local obstruction that contradicts extremality. This step adapts the joint-range perturbation idea of \citet{Pensia2025BinaryTesting}: we rule out such a column by producing two feasible induced pairs whose midpoint is the original pair.

\subsection{Convex hull of SPR mechanisms}
\label{sec:geom-convex-hull}

Once the active columns are LR-contiguous, the staircase channel decomposes into a convex combination of SPR mechanisms.

\begin{proposition}[Convex hull of SPR channels]
\label{prop:convex-hull-partition-rr}
Let $m\ge 1$, and let $Q\in\mathcal Q_{\eps,m}$ be a staircase channel. Assume every active column of $Q$ has an LR-contiguous high set, with active constant columns represented by the high set $[k]$. Then there exist LR-contiguous partitions $\pi_1,\dots,\pi_N$, channels $\widehat Q^{(1)},\dots,\widehat Q^{(N)}\in\mathcal Q_{\eps,m}$, and coefficients $\beta_1,\dots,\beta_N>0$ with $\sum_{t=1}^N \beta_t=1$ such that
\begin{align}
    Q=\sum_{t=1}^N \beta_t \widehat Q^{(t)},
\end{align}
where, for each $t$, the channel $\widehat Q^{(t)}$ is obtained from the SPR mechanism $Q^{\pi_t}$ by relabeling its active outputs into $[m]$ and padding the remaining columns with zeros.
\end{proposition}

The proof is deferred to Appendix~\ref{sec:proof-convex-hull-partition-rr}. It rewrites $Q$ as a weighted family of LR-contiguous intervals, uses the row-sum condition to show that every row sees the same total interval weight, and then peels this weighted interval system into layers, each layer being a partition of $[k]$. Each layer yields an SPR mechanism on the same output alphabet after relabeling and zero-padding. This peeling step can be viewed as a weighted flow-decomposition argument \cite{Ahuja1993NetworkFlows} specialized to interval families.

\subsection{Nontrivial coarsenings are not extreme}
\label{sec:geom-coarsening}

Finally, we show that deterministically merging active outputs of an SPR mechanism does not create new extreme points.

\begin{proposition}[Nontrivial coarsenings are not extreme]
\label{prop:no-new-vertices-under-coarsening}
Let $\pi=(B_1,\dots,B_s)$ be an LR-contiguous partition of $[k]$, let $Q^\pi\in\mathcal Q_{\eps,s}$ be the associated SPR mechanism, and let $\tau:[s]\to[\ell]$ be deterministic. If $t:=|\tau([s])|$ satisfies $1<t<s$, then the image point
\begin{align}
    \bigl((\tau\circ Q^\pi)^\top p,(\tau\circ Q^\pi)^\top q\bigr)\in\mathcal R_\ell(p,q)
\end{align}
is not an extreme point of $\mathcal R_\ell(p,q)$.
\end{proposition}

The proof is deferred to Appendix~\ref{sec:proof-no-new-vertices-under-coarsening}. The key observation is that such a coarsened image decomposes into a randomized-response point on the coarser partition induced by $\tau$, together with a hypothesis-independent remainder coming from an input-independent channel. This gives a nontrivial convex decomposition inside $\mathcal R_\ell(p,q)$.

\subsection{Proof of the SPR theorem and its \texorpdfstring{$f$-divergence}{f-divergence} corollary}
\label{sec:geom-proof}

Putting these three propositions together gives the main geometric statement: Theorem~\ref{thm:extreme-points-spr}.

\begin{proof}[Proof of Theorem~\ref{thm:extreme-points-spr}]
Let $(u,v)\in\mathcal R_\ell(p,q)$ be extreme. By Proposition~\ref{prop:extreme-refinement}, there exist $m\ge 1$, a staircase channel $\widetilde Q\in\mathcal Q_{\eps,m}$ whose active columns are LR-contiguous, and a deterministic map $\tau:[m]\to[\ell]$ such that $(u,v)=\bigl((\tau\circ\widetilde Q)^\top p,(\tau\circ\widetilde Q)^\top q\bigr)$.
By Proposition~\ref{prop:convex-hull-partition-rr}, $\widetilde Q=\sum_{t=1}^N \beta_t \widehat Q^{(t)}$, with $\beta_t>0$ and $\sum_{t=1}^N \beta_t=1$, where each $\widehat Q^{(t)}$ is a relabeled and zero-padded SPR mechanism. Applying $\tau$ and taking induced marginals gives
\begin{align}
    (u,v)=\sum_{t=1}^N
    \beta_t\bigl((\tau\circ\widehat Q^{(t)})^\top p,(\tau\circ\widehat Q^{(t)})^\top q\bigr).
\end{align}
Since $(u,v)$ is extreme, every term with $\beta_t>0$ must already equal $(u,v)$. Fix such a $t$.

If $\tau\circ\widehat Q^{(t)}$ has only one active output, then its induced pair is realized by the one-block partition $\pi=([k])$, viewed as an $\ell$-output channel by relabeling its unique active output and padding zeros. Otherwise $\tau\circ\widehat Q^{(t)}$ has more than one active output. Since its induced pair is $(u,v)$ and hence extreme, Proposition~\ref{prop:no-new-vertices-under-coarsening} implies that $\tau$ does not merge two active outputs of $\widehat Q^{(t)}$. Thus $\tau\circ\widehat Q^{(t)}$ differs from $\widehat Q^{(t)}$ only by relabeling active outputs and padding zeros. In either case, $(u,v)$ is induced by an SPR mechanism, up to relabeling and zero-padding.
\end{proof}

The extreme-point characterization converts into optimal mechanisms, both with and without a communication constraint, via Corollary~\ref{cor:fdiv-spr-optimality}.

\begin{proof}[Proof of Corollary~\ref{cor:fdiv-spr-optimality}]
    Fix any $\ell'\ge 1$ and any $Q\in\mathcal Q_{\eps,\ell'}$, and set $(u,v)=(Q^\top p,Q^\top q)\in\mathcal R_{\ell'}(p,q)$. Since $\mathcal R_{\ell'}(p,q)$ is a compact convex polytope, write $(u,v)$ as a convex combination of its extreme points: $(u,v)=\sum_{t=1}^N \lambda_t(u^{(t)},v^{(t)})$ for $\lambda_t\ge 0$, $\sum_{t=1}^N\lambda_t=1$.
    By convexity of $D_f$,
    \begin{align}
        D_f(u\|v)
        \le \sum_{t=1}^N \lambda_t D_f(u^{(t)}\|v^{(t)})
        \le \max_t D_f(u^{(t)}\|v^{(t)}),
    \end{align}
so some extreme point $(u^{(t)},v^{(t)})$ satisfies $D_f(u^{(t)}\|v^{(t)}) \ge D_f(u\|v)$. By Theorem~\ref{thm:extreme-points-spr}, that extreme point is induced by an SPR mechanism with at most $\min(\ell',k)$ active outputs.

For the unconstrained claim, take $\ell'$ arbitrary: every finite-output $\eps$-LDP mechanism is matched or improved by an SPR mechanism with at most $k$ active outputs, and any such SPR mechanism is itself a feasible finite-output mechanism. Hence the unconstrained finite-output optimum is attained among SPR mechanisms with at most $k$ active outputs. For the output-constrained claim, take $\ell'=\ell$: every $\ell$-output mechanism is matched or improved by an SPR mechanism with at most $\ell$ active outputs, and any such SPR mechanism can be viewed as an element of $\mathcal Q_{\eps,\ell}$ by padding zero columns.
\end{proof}

%% file: sections/05_dynamic_program.tex
\section{Exact dynamic program}
\label{sec:dynamic-program}
\label{sec:polynomial-time-algorithm}

Fix an $f$-divergence objective and an output budget $1\le \ell\le k$. By Corollary~\ref{cor:fdiv-spr-optimality}, it is enough to optimize over LR-contiguous partitions of $[k]$ into at most $\ell$ nonempty blocks. The unconstrained finite-output problem is the case $\ell=k$. We now use this partition structure to obtain an $O(\ell k^2)$ dynamic program, and hence an $O(k^3)$ algorithm in the unconstrained case. Section~\ref{sec:pseudocode} gives a pseudocode version.

\subsection{Interval scores and DP recurrence}
\label{sec:dp-interval-scores}
\label{sec:dp-recurrence}

For $z\in\mathbb{R}_+^k$, let $\displaystyle \mu(z):=(q^\top z)\,f\!\left(\frac{p^\top z}{q^\top z}\right)$ be the column score associated with the chosen $f$-divergence. For $1\le a\le b\le k$, let
\begin{align}
    \mu[a:b]
    :=
    \mu\!\left(
        \bigl(1+(e^\varepsilon-1)\mathbf{1}\{a\le i\le b\}\bigr)_{i=1}^k
    \right).
\end{align}
Thus $\mu[a:b]$ is the contribution of the unnormalized staircase pattern with high set $[a:b]$.

If $\pi=(B_1,\dots,B_s)$ is an LR-contiguous partition of $[k]$ into $s$ nonempty blocks, then each output column of $Q^\pi$ is the corresponding unnormalized staircase pattern divided by $e^\varepsilon+s-1$. Therefore
\begin{align}
    D_f\bigl((Q^\pi)^\top p\|(Q^\pi)^\top q\bigr)
    =
    \frac{\sum_{t=1}^s \mu(B_t)}{e^\varepsilon+s-1},
\end{align}
where $\mu(B_t)=\mu[a:b]$ when $B_t=[a:b]$. Thus, for each fixed $s\in\{1,\dots,\ell\}$, it is enough to maximize the numerator over all partitions of $[k]$ into $s$ nonempty contiguous blocks. This yields Proposition~\ref{prop:dp-contiguous-partition}.

\begin{proof}[Proof of Proposition~\ref{prop:dp-contiguous-partition}]
The base case is immediate. For the recurrence, take an optimal partition of $[1:i]$ into $s$ contiguous blocks, and let its last block be $[t+1:i]$. Then the prefix $[1:t]$ must already be partitioned optimally into $s-1$ contiguous blocks; otherwise we could replace it by a better one and improve the whole partition. This gives the recurrence. The last equation follows because for fixed $s$ the denominator $e^\varepsilon+s-1$ is constant.
\end{proof}

\subsection{Runtime, output constraint, and traceback}
\label{sec:dp-runtime-traceback}

\begin{proof}[Proof of Corollary~\ref{cor:dp-runtime}]
Sorting the input alphabet by likelihood ratio takes $O(k \log k)$ time, which is dominated by the $O(\ell k^2)$ DP fill. Let $P_i:=\sum_{j=1}^i p_j$ and $Q_i:=\sum_{j=1}^i q_j$ with $P_0=Q_0=0$. Then for any interval $[a:b]$,
\begin{align}
    \mu[a:b]
    =
    \left(1+(e^\varepsilon-1)(Q_b-Q_{a-1})\right)
    f\!\left(
        \frac{1+(e^\varepsilon-1)(P_b-P_{a-1})}
             {1+(e^\varepsilon-1)(Q_b-Q_{a-1})}
    \right).
\end{align}
Thus each interval score can be evaluated in $O(1)$ time from the prefix sums, without precomputing or storing the $O(k^2)$ possible scores.
The DP table $F[s,i]$ has $O(\ell k)$ states, and each state checks $O(k)$ split points, evaluating the needed interval score on demand from the prefix sums. Hence filling the table takes $O(\ell k^2)$ time. The prefix sums use $O(k)$ space and the DP table uses $O(\ell k)$ space, so the total space is $O(\ell k)$.

To recover an optimal mechanism, first choose
\begin{align}
    s^\star\in\arg\max_{1\le s\le \ell}
    \frac{F[s,k]}{e^\varepsilon+s-1}.
\end{align}
Then trace backward from $(s^\star,k)$: at state $(s,i)$, choose any
\begin{align}
    t\in\arg\max_{s-1\le t<i}
    \Bigl\{F[s-1,t]+\mu[t+1:i]\Bigr\},
\end{align}
record the block $[t+1:i]$, and continue from $(s-1,t)$. The final block is $[1:i]$. This traceback uses only the stored $F$ table and the prefix sums, so it does not change the $O(\ell k)$ space bound. It adds at most $O(\ell k)$ time, since at most $\ell$ states are traced back and each recomputes one maximization over $O(k)$ split points. Thus the dominant running-time cost remains the $O(\ell k^2)$ DP fill, which becomes $O(k^3)$ when $\ell=k$.
\end{proof}

\subsection{Pseudocode}
\label{sec:pseudocode}

\begin{algorithm}[H]
\small
\DontPrintSemicolon
\KwIn{Distributions $p, q \in \Delta^k$, privacy level $\varepsilon > 0$, output budget $1 \le \ell \le k$, convex generator $f$ with $f(1)=0$.}
\KwOut{Optimal value $V^\star$, optimal LR-contiguous partition $\pi^\star = (B_1, \dots, B_{s^\star})$, and the associated SPR mechanism $Q^{\pi^\star}$.}

\BlankLine
\tcp{Stage 1: sort by likelihood ratio and compute prefix sums}
Let $\sigma$ be a permutation of $[k]$ with $p_{\sigma(1)}/q_{\sigma(1)} \le \cdots \le p_{\sigma(k)}/q_{\sigma(k)}$\;
Relabel so that $p_i \leftarrow p_{\sigma(i)}$ and $q_i \leftarrow q_{\sigma(i)}$ for $i = 1, \dots, k$\;
Set $P_0 \leftarrow 0$, $Q_0 \leftarrow 0$\;
\For{$i = 1, \dots, k$}{
    $P_i \leftarrow P_{i-1} + p_i$\;
    $Q_i \leftarrow Q_{i-1} + q_i$\;
}

\BlankLine
\tcp{Interval-score oracle: $O(1)$ from prefix sums}
\SetKwFunction{FMu}{\textsc{IntervalScore}}
\SetKwProg{Fn}{Function}{:}{}
\Fn{\FMu{$a, b$}}{
    $\Delta P \leftarrow P_b - P_{a-1}$\;
    $\Delta Q \leftarrow Q_b - Q_{a-1}$\;
    \Return $\bigl(1 + (e^\varepsilon - 1)\Delta Q\bigr)\, f\!\left(\dfrac{1 + (e^\varepsilon - 1)\Delta P}{1 + (e^\varepsilon - 1)\Delta Q}\right)$\;
}

\BlankLine
\tcp{Stage 2: fill the DP table $F[s, i]$}
Initialize $F[s, i] \leftarrow -\infty$ for all $1 \le s \le \ell$, $1 \le i \le k$\;
\For{$i = 1, \dots, k$}{
    $F[1, i] \leftarrow \FMu(1, i)$\;
}
\For{$s = 2, \dots, \ell$}{
    \For{$i = s, \dots, k$}{
        $F[s, i] \leftarrow \displaystyle\max_{s-1 \le t < i}\bigl\{F[s-1, t] + \FMu(t+1, i)\bigr\}$\;
    }
}

\BlankLine
\tcp{Stage 3: select the optimal block count and value}
$s^\star \leftarrow \displaystyle\argmax_{1 \le s \le \ell} \dfrac{F[s, k]}{e^\varepsilon + s - 1}$\;
$V^\star \leftarrow \dfrac{F[s^\star, k]}{e^\varepsilon + s^\star - 1}$\;

\BlankLine
\tcp{Stage 4: traceback to recover the partition (in sorted order)}
$i \leftarrow k$\;
\For{$s = s^\star, s^\star - 1, \dots, 2$}{
    $t^\star \leftarrow \displaystyle\argmax_{s-1 \le t < i}\bigl\{F[s-1, t] + \FMu(t+1, i)\bigr\}$\;
    $\widetilde{B}_s \leftarrow [t^\star + 1 : i]$\;
    $i \leftarrow t^\star$\;
}
$\widetilde{B}_1 \leftarrow [1 : i]$\;

\BlankLine
\tcp{Undo the sort to express blocks on the original alphabet}
\For{$s = 1, \dots, s^\star$}{
    $B_s \leftarrow \{\sigma(j) : j \in \widetilde{B}_s\}$\;
}
$\pi^\star \leftarrow (B_1, \dots, B_{s^\star})$\;

\BlankLine
\tcp{Stage 5: instantiate the SPR mechanism}
\For{$s \in [s^\star]$, $x \in [k]$}{
    $Q^{\pi^\star}(s \mid x) \leftarrow \dfrac{1 + (e^\varepsilon - 1)\mathbf{1}\{x \in B_s\}}{e^\varepsilon + s^\star - 1}$\;
}

\Return $V^\star$, $\pi^\star$, $Q^{\pi^\star}$\;
\caption{Optimal SPR mechanism via dynamic programming}
\label{alg:optimal-spr-dp}
\end{algorithm}
\FloatBarrier

%% file: sections/06_conclusion.tex
\section{Conclusion and limitations}
\label{sec:conclusion-and-limitations}
We showed that the exact mechanism-design problem for binary testing under local privacy has a much smaller structure than the general staircase formulation suggests. After sorting by likelihood ratio, optimal mechanisms can be chosen by partitioning the ordered alphabet into contiguous blocks and applying randomized response to the block label (Section~\ref{sec:geometry-of-spr}), yielding an exact polynomial-time dynamic program (Section~\ref{sec:dynamic-program}). For $E_\gamma$-divergences this collapses to a closed-form binary mechanism, and the same machinery extends to R\'enyi divergences and $\ell_r$ distance objectives (Sections~\ref{sec:main-results-egamma}--\ref{sec:main-results-other}).

Our characterization is specific to simple binary hypothesis testing under pure, non-interactive $\varepsilon$-LDP on a finite alphabet. Extending the SPR geometry to multi-hypothesis or composite testing, approximate or R\'enyi local privacy, and interactive protocols remains open. On the algorithmic side, the $O(\ell k^2)$ dynamic program relies on the additive block decomposition of $f$-divergences; finding comparably efficient exact algorithms for broader quasi-convex objectives, such as Chernoff information, is an interesting direction. Another natural question is how to adapt the mechanism-design results when the hypotheses $P_0,P_1$ are not known exactly and must be estimated or specified through uncertainty classes.

%% file: appendix/A_extreme_refinements/A0_overview.tex
\section{Extreme points admit LR-contiguous staircase refinements}
\label{sec:proof-LR-contiguous}

We prove Proposition~\ref{prop:extreme-refinement} under the assumption that inputs with the same likelihood ratios have been merged, i.e., that the likelihood ratios are strictly increasing in the LR order.

Let $(u,v)\in\mathcal R_\ell(p,q)$ be an extreme point. We will show that $(u,v)$ is realized, up to deterministic post-processing, by a staircase channel whose active columns are LR-contiguous.

The proof has three steps. First, we refine an arbitrary channel realizing $(u,v)$ into a staircase channel. Second, we show that if one of its active columns is not LR-contiguous, then on some LR-ordered triple it creates a forbidden local configuration. Third, we use a perturbation argument to show that each such configuration yields a nontrivial midpoint decomposition of $(u,v)$, contradicting extremality.

Figure~\ref{fig:proof-logic} traces the logical flow of the proof: starting from an extreme image pair, we refine to a staircase (Lemma~\ref{lem:nested-staircase-refinement}), assume for contradiction that some active column has a non-contiguous high set, apply the dichotomy of Lemma~\ref{lemma:any-101-triple-forces-zigzag-or-3-cycle}, perturb in either branch (Lemmas~\ref{lemma:zigzag-implies-non-extremality} and~\ref{lemma:3-cycle-implies-non-extremality-in-joint-range}), push the perturbations back through $\tau$, and reach a contradiction with extremality.

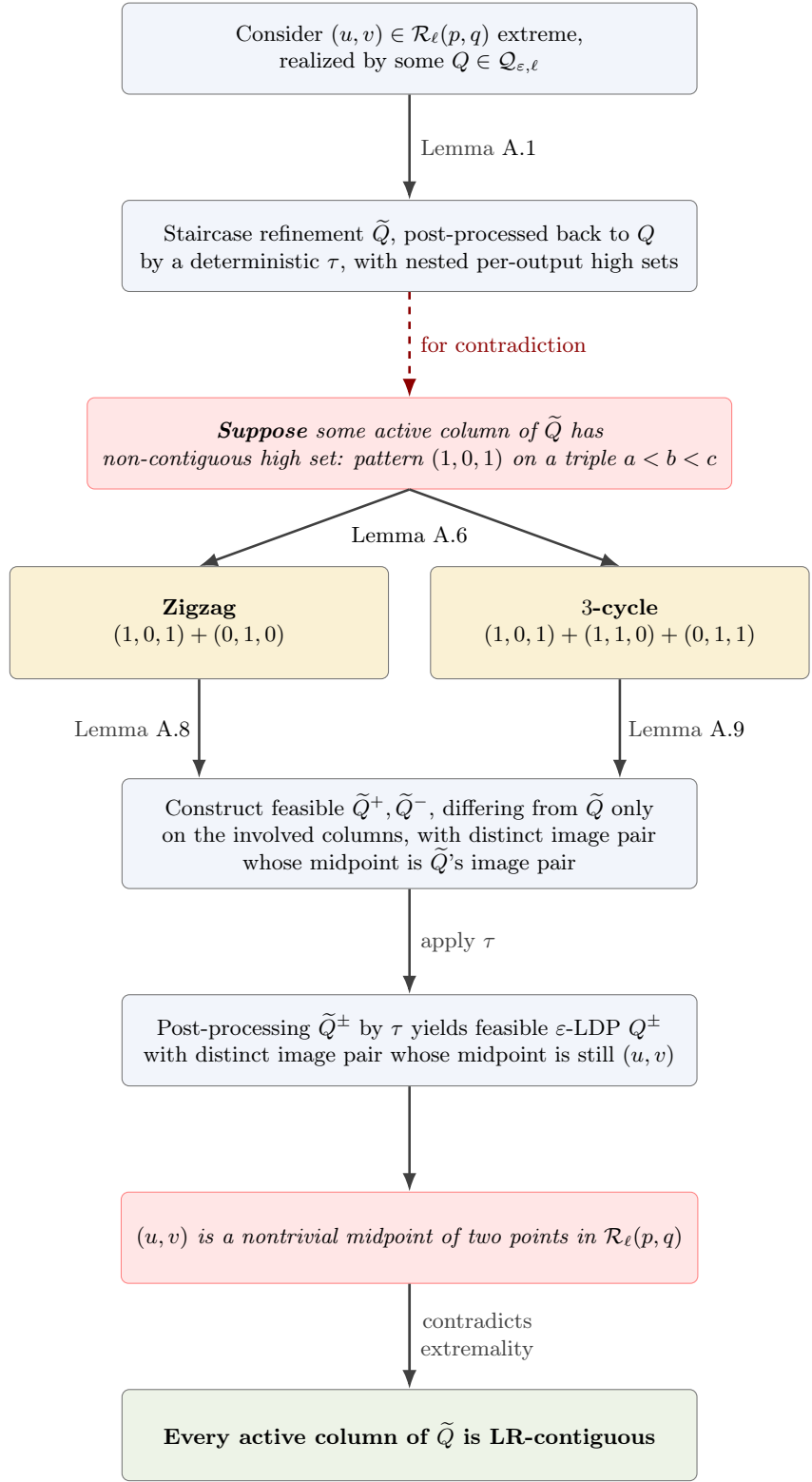
\begin{figure}[tbp]
  \centering
  \resizebox{!}{0.92\textheight}{%
    \input{figures/proof_logic_figure}%
  }
  \caption{Logical structure of the proof of Proposition~\ref{prop:extreme-refinement}. Blue boxes are established facts; the red dashed arrow marks the assumption-for-contradiction trail; yellow boxes are the two forbidden configurations; the green box is the conclusion.}
  \label{fig:proof-logic}
\end{figure}

\subsection{Staircase refinement}
\label{sec:staircase-refinement}
The next lemma shows that any fixed-output $\varepsilon$-LDP channel can be refined into a staircase channel that merges back to the original one under a deterministic post-processing map. This lets us work at the staircase level and later transfer a contradiction back to the original image point. Example~\ref{ex:staircase-refinement} below illustrates the construction on a small $5\times 2$ channel.
\begin{lemma}[Nested staircase refinement]
\label{lem:nested-staircase-refinement}
Let $Q\in\mathbb{R}_+^{k\times \ell}$ be a row-stochastic $\varepsilon$-LDP channel, where $\varepsilon>0$. Then there exist a staircase $\varepsilon$-LDP channel $\widetilde Q\in\mathbb{R}_+^{k\times \widetilde m}$ with output alphabet $\widetilde{\mathcal Y}=[\widetilde m]$, and a deterministic map $\tau:\widetilde{\mathcal Y}\to[\ell]$ such that
\begin{align*}
    Q(i,y)=\sum_{\tilde y:\tau(\tilde y)=y}\widetilde Q(i,\tilde y)
    \qquad (i\in[k],\ y\in[\ell]).
\end{align*}
Moreover, for each $y\in[\ell]$, the high sets of the fine columns merged into $y$ are nested.
\end{lemma}

\begin{proof}
Zero columns need no refinement, so fix a nonzero column $Q(\cdot,y)$ of $Q$.
Because $Q$ is $\varepsilon$-LDP and this column is nonzero, every entry in it is positive.
Let
\begin{align*}
    \theta:=\min_{i\in[k]} Q(i,y)>0
\end{align*}
be its smallest entry. Then every entry in this column lies between $\theta$ and $e^\varepsilon\theta$.

We now show that this column can be written as a sum of staircase columns with nested high sets.

Write each entry as
\begin{align*}
    Q(i,y)=\theta\bigl(1+(e^\varepsilon-1)u_i\bigr), \qquad u_i\in[0,1],
\end{align*}
where
\begin{align*}
    u_i:=\frac{Q(i,y)-\theta}{(e^\varepsilon-1)\theta}\in[0,1].
\end{align*}
Here, $u_i$ records how much of the maximum possible ``extra above the minimum'' row $i$ receives.

If $u_i=0$ for all $i$, then the column is already constant: $Q(i,y)=\theta$ for all $i$.
In that case we define a single fine column by
\begin{align*}
    \widetilde Q(i,(y,0)):=\theta \qquad (i\in[k]),
\end{align*}
and there is nothing further to do for this $y$.

Otherwise, let
\begin{align*}
    0=\lambda_0<\lambda_1<\cdots<\lambda_L\le 1
\end{align*}
be the distinct values taken by the $u_i$'s. For each level $t=1,\dots,L$, let
\begin{align*}
    H_t:=\{i\in[k]:u_i\ge \lambda_t\}.
\end{align*}
Then
\begin{align*}
    H_1\supseteq H_2\supseteq \cdots \supseteq H_L,
\end{align*}
because rows with a larger boost automatically belong to all smaller-boost groups.

Each $u_i$ is obtained by adding the jump sizes $\lambda_t-\lambda_{t-1}$ for exactly those sets $H_t$ that contain row $i$:
\begin{align*}
    u_i=\sum_{t=1}^L (\lambda_t-\lambda_{t-1})\mathbf 1\{i\in H_t\}.
\end{align*}
Now let $H_0=\varnothing$, set
\begin{align*}
    \alpha_t:=\theta(\lambda_t-\lambda_{t-1}) \qquad (t=1,\dots,L),
\end{align*}
and
\begin{align*}
    \alpha_0:=\theta(1-\lambda_L).
\end{align*}
Substituting the expression for $u_i$ into
$Q(i,y)=\theta\bigl(1+(e^\varepsilon-1)u_i\bigr)$ gives
\begin{align*}
    Q(i,y)=\sum_{t=0}^L \alpha_t\Bigl(1+(e^\varepsilon-1)\mathbf 1\{i\in H_t\}\Bigr).
\end{align*}
For this fixed $y$, define the corresponding fine columns by
\begin{align*}
    \widetilde Q(i,(y,t))
    :=\alpha_t\Bigl(1+(e^\varepsilon-1)\mathbf 1\{i\in H_t\}\Bigr)
    \qquad (t=0,\dots,L).
\end{align*}
Then
\begin{align*}
    Q(i,y)=\sum_{t=0}^L \widetilde Q(i,(y,t))
    \qquad (i\in[k]).
\end{align*}

Apply the same construction to every nonzero column of $Q$, using for each column $y$ its own levels $\lambda_t^{(y)}$, sets $H_t^{(y)}$, and coefficients $\alpha_t^{(y)}$.
The resulting fine columns form the refined channel $\widetilde Q$.
Let $\tau$ send each fine column back to the original column it came from.
Thus,
\begin{align*}
    Q(i,y)=\sum_{\tilde y:\tau(\tilde y)=y}\widetilde Q(i,\tilde y) \qquad (i\in[k],\ y\in[\ell]).
\end{align*}

Each fine column has the form $\alpha\bigl(1+(e^\varepsilon-1)\mathbf 1\{i\in H\}\bigr)$ for some $\alpha\ge 0$ and some set $H\subseteq[k]$, so it has staircase form.
Since within each original column the fine columns sum back to that column, summing over all original columns shows that each row sum of $\widetilde Q$ equals the corresponding row sum of $Q$; hence $\widetilde Q$ is row-stochastic. Finally, for a fixed original output $y$, the fine columns coming from $Q(\cdot,y)$ have high sets among the sets $H_t^{(y)}$, and these sets are nested.
\end{proof}
\begin{remark}
In the proof of Lemma~\ref{lem:nested-staircase-refinement}, we allow constant fine columns to be written with high set $\emptyset$, since this makes the decomposition clearer.
Elsewhere in the paper, when discussing staircase channels abstractly, we use the convention that an active constant column is represented with high set $[k]$. These are just two equivalent representations of the same constant column, and the choice here is purely for notational convenience.
\end{remark}

\begin{example}[Staircase refinement of a $5\times 2$ channel]
\label{ex:staircase-refinement}
Take $\eps=1$, so $e^\eps=e$, and consider the $5\times 2$ mechanism
\begin{align*}
    Q \approx
    \begin{bmatrix}
    0.3124 & 0.6876\\
    0.4017 & 0.5983\\
    0.2865 & 0.7135\\
    0.3588 & 0.6412\\
    0.4279 & 0.5721
    \end{bmatrix}.
\end{align*}
Up to rounding, each row sums to $1$ and each column has max/min ratio at most $e^\eps$, so $Q$ is $\eps$-LDP; this is a generic non-staircase example. We decompose each output column into staircase pieces with nested high sets. For the first output,
\begin{align*}
    Q(\cdot,y_1)
    \approx
    \begin{bmatrix}
    0.3124\\
    0.4017\\
    0.2865\\
    0.3588\\
    0.4279
    \end{bmatrix}
    =
    0.2042
    \begin{bmatrix}
    1\\1\\1\\1\\1
    \end{bmatrix}
    +
    0.0151
    \begin{bmatrix}
    e^\eps\\ e^\eps\\ 1\\ e^\eps\\ e^\eps
    \end{bmatrix}
    +
    0.0270
    \begin{bmatrix}
    1\\ e^\eps\\ 1\\ e^\eps\\ e^\eps
    \end{bmatrix}
    +
    0.0250
    \begin{bmatrix}
    1\\ e^\eps\\ 1\\ 1\\ e^\eps
    \end{bmatrix}
    +
    0.0152
    \begin{bmatrix}
    1\\ 1\\ 1\\ 1\\ e^\eps
    \end{bmatrix}.
\end{align*}
The high sets are
\begin{align*}
    \{1,2,4,5\}\supset \{2,4,5\}\supset \{2,5\}\supset \{5\}.
\end{align*}
For the second output,
\begin{align*}
    Q(\cdot,y_2)
    \approx
    \begin{bmatrix}
    0.6876\\
    0.5983\\
    0.7135\\
    0.6412\\
    0.5721
    \end{bmatrix}
    =
    0.4898
    \begin{bmatrix}
    1\\1\\1\\1\\1
    \end{bmatrix}
    +
    0.0152
    \begin{bmatrix}
    e^\eps\\ e^\eps\\ e^\eps\\ e^\eps\\ 1
    \end{bmatrix}
    +
    0.0250
    \begin{bmatrix}
    e^\eps\\ 1\\ e^\eps\\ e^\eps\\ 1
    \end{bmatrix}
    +
    0.0270
    \begin{bmatrix}
    e^\eps\\ 1\\ e^\eps\\ 1\\ 1
    \end{bmatrix}
    +
    0.0151
    \begin{bmatrix}
    1\\ 1\\ e^\eps\\ 1\\ 1
    \end{bmatrix}.
\end{align*}
The high sets are
\begin{align*}
    \{1,2,3,4\}\supset \{1,3,4\}\supset \{1,3\}\supset \{3\}.
\end{align*}
\paragraph{Refined channel.}
Stacking the staircase columns from both outputs gives
\begin{align*}
    \widetilde Q \approx [\,A \mid B\,],
\end{align*}
where $A$ holds the $y_1$ staircase columns and $B$ holds the $y_2$ staircase columns:
\begin{align*}
    A \approx
    \begin{bmatrix}
    0.2042 & 0.0151\,e^\eps & 0.0270 & 0.0250 & 0.0152 \\
    0.2042 & 0.0151\,e^\eps & 0.0270\,e^\eps & 0.0250\,e^\eps & 0.0152 \\
    0.2042 & 0.0151 & 0.0270 & 0.0250 & 0.0152 \\
    0.2042 & 0.0151\,e^\eps & 0.0270\,e^\eps & 0.0250 & 0.0152 \\
    0.2042 & 0.0151\,e^\eps & 0.0270\,e^\eps & 0.0250\,e^\eps & 0.0152\,e^\eps
    \end{bmatrix},
\end{align*}
\begin{align*}
    B \approx
    \begin{bmatrix}
    0.4898 & 0.0152\,e^\eps & 0.0250\,e^\eps & 0.0270\,e^\eps & 0.0151 \\
    0.4898 & 0.0152\,e^\eps & 0.0250 & 0.0270 & 0.0151 \\
    0.4898 & 0.0152\,e^\eps & 0.0250\,e^\eps & 0.0270\,e^\eps & 0.0151\,e^\eps \\
    0.4898 & 0.0152\,e^\eps & 0.0250\,e^\eps & 0.0270 & 0.0151 \\
    0.4898 & 0.0152 & 0.0250 & 0.0270 & 0.0151
    \end{bmatrix}.
\end{align*}
Summing the columns of $A$ recovers $Q(\cdot,y_1)$, and summing the columns of $B$ recovers $Q(\cdot,y_2)$, so $Q$ is obtained from $\widetilde{Q}$ by deterministic post-processing.
\end{example}

\subsection{Staircase obstructions}
\label{sec:staircase-obstructions}
We now identify the local staircase configurations that obstruct LR-contiguity. We first show that if a staircase refinement contains a noncontiguous pattern $(101)$ on some LR-ordered triple, then on that same triple it must contain one of two forbidden configurations. In the next subsection, we show that either configuration yields a nontrivial midpoint decomposition and hence contradicts extremality.
\begin{definition}[The Zigzag Configuration]
Let $Q\in\mathbb{R}_+^{k\times m}$ be an $\varepsilon$-LDP \emph{staircase} channel with $\varepsilon>0$, and assume the inputs $[k]$ are ordered in strictly increasing likelihood ratio. We say that $Q$ \emph{contains a zigzag configuration} if there exist three (LR-ordered) rows $a<b<c$ and two distinct columns $j\neq j'$ such that for some $\theta_j,\theta_{j'}>0$,
\begin{align*}
    (Q(a,j),Q(b,j),Q(c,j))=(e^\varepsilon\theta_j,\ \theta_j,\ e^\varepsilon\theta_j),
\end{align*}
and
\begin{align*}
    (Q(a,j'),Q(b,j'),Q(c,j'))=(\theta_{j'},\ e^\varepsilon\theta_{j'},\ \theta_{j'}).
\end{align*}
Equivalently, restricted to rows $\{a,b,c\}$, column $j$ has pattern $(101)$ and column $j'$ has pattern $(010)$.
\end{definition}
\paragraph{In matrix form.}
The entries shown sit at rows $a<b<c$ and columns $j,j'$ as in the definition above; remaining entries (denoted $\cdots$) are arbitrary:
\begin{align*}
    Q=
    \begin{bmatrix}
    \cdots & \cdots & \cdots & \cdots & \cdots \\
    \cdots & e^\varepsilon \theta_j & \cdots & \theta_{j'} & \cdots \\
    \cdots & \cdots & \cdots & \cdots & \cdots \\
    \cdots & \theta_j & \cdots & e^\varepsilon \theta_{j'} & \cdots \\
    \cdots & \cdots & \cdots & \cdots & \cdots \\
    \cdots & e^\varepsilon \theta_j & \cdots & \theta_{j'} & \cdots \\
    \cdots & \cdots & \cdots & \cdots & \cdots
    \end{bmatrix}.
\end{align*}
\begin{definition}[The $3$-Cycle Configuration]
Let $Q\in\mathbb{R}_+^{k\times m}$ be an $\varepsilon$-LDP \emph{staircase} channel with $\varepsilon>0$, and assume the inputs $[k]$ are ordered in strictly increasing likelihood ratio. We say that $Q$ \emph{contains a $3$-cycle configuration} if there exist three (LR-ordered) rows $a<b<c$ and three distinct columns $j_1,j_2,j_3$ such that for some $\theta_{j_1},\theta_{j_2},\theta_{j_3}>0$,
\begin{align*}
    (Q(a,j_1),Q(b,j_1),Q(c,j_1))=(e^\varepsilon\theta_{j_1},\ \theta_{j_1},\ e^\varepsilon\theta_{j_1}),
\end{align*}
\begin{align*}
    (Q(a,j_2),Q(b,j_2),Q(c,j_2))=(e^\varepsilon\theta_{j_2},\ e^\varepsilon\theta_{j_2},\ \theta_{j_2}),
\end{align*}
and
\begin{align*}
    (Q(a,j_3),Q(b,j_3),Q(c,j_3))=(\theta_{j_3},\ e^\varepsilon\theta_{j_3},\ e^\varepsilon\theta_{j_3}).
\end{align*}
Equivalently, restricted to rows $\{a,b,c\}$, the three columns realize the three patterns $(101)$, $(110)$, and $(011)$ (in any column order).
\end{definition}
\paragraph{In matrix form.}
The entries shown sit at rows $a<b<c$ and columns $j_1,j_2,j_3$ as in the definition above; remaining entries (denoted $\cdots$) are arbitrary:
\begin{align*}
    Q=
    \begin{bmatrix}
    \cdots & \cdots & \cdots & \cdots & \cdots & \cdots & \cdots \\
    \cdots & e^\varepsilon \theta_{j_1} & \cdots & e^\varepsilon \theta_{j_2} & \cdots & \theta_{j_3} & \cdots \\
    \cdots & \cdots & \cdots & \cdots & \cdots & \cdots & \cdots \\
    \cdots & \theta_{j_1} & \cdots & e^\varepsilon \theta_{j_2} & \cdots & e^\varepsilon \theta_{j_3} & \cdots \\
    \cdots & \cdots & \cdots & \cdots & \cdots & \cdots & \cdots \\
    \cdots & e^\varepsilon \theta_{j_1} & \cdots & \theta_{j_2} & \cdots & e^\varepsilon \theta_{j_3} & \cdots \\
    \cdots & \cdots & \cdots & \cdots & \cdots & \cdots & \cdots
    \end{bmatrix}.
\end{align*}
To prove LR-contiguity, it is enough to rule out the local pattern $(101)$ on LR-ordered triples. Indeed, a staircase column has noncontiguous high set if and only if there exist rows $a<b<c$ on which the column has pattern $(101)$. The next lemma shows that any such $(101)$ pattern forces a more structured forbidden configuration on the same triple.
\begin{lemma}[Any $101$ triple forces a zigzag configuration or $3$-cycle configuration on that triple]
\label{lemma:any-101-triple-forces-zigzag-or-3-cycle}
Let $Q \in \mathbb{R}_+^{k\times m}$ be a row-stochastic staircase $\varepsilon$-LDP channel with $\varepsilon>0$, and fix a staircase representation of its columns: for each $y\in[m]$, choose $\theta_y \ge 0$ and $H_y\subseteq[k]$ such that
\begin{align*}
    Q(i,y)=\theta_y\Bigl(1+(e^\varepsilon-1)\mathbf{1}\{i\in H_y\}\Bigr).
\end{align*}
Fix rows $a < b < c$. For each $s \in \{0,1\}^3$, let
\begin{align*}
    w_s := \sum_{y=1}^{m} \theta_y
    \mathbf{1}\left\{
    \bigl(\mathbf{1}\{a\in H_y\},\mathbf{1}\{b\in H_y\},\mathbf{1}\{c\in H_y\}\bigr)=s
    \right\}.
\end{align*}

If $w_{101} > 0$, then either $w_{010} > 0$, or both $w_{110} > 0$ and $w_{011} > 0$.
\end{lemma}

\begin{proof}
For $i\in\{a,b,c\}$, define
\begin{align*}
    h_i:=\sum_{y:i\in H_y}\theta_y.
\end{align*}
Since $Q$ is row-stochastic,
\begin{align*}
    1=\sum_{y=1}^m Q(i,y) = \sum_{y=1}^m \theta_y + (e^\varepsilon-1)\sum_{y:i\in H_y}\theta_y,
\end{align*}
so $h_a=h_b=h_c$.

Also,
\begin{align*}
    h_a = w_{100}+w_{101}+w_{110}+w_{111},
\end{align*}
\begin{align*}
    h_b = w_{010}+w_{011}+w_{110}+w_{111},
\end{align*}
\begin{align*}
    h_c = w_{001}+w_{011}+w_{101}+w_{111}.
\end{align*}
Hence
\begin{align}
    w_{100}+w_{101} &= w_{010}+w_{011}, \label{eq:ha-eq-hb}\\
    w_{001}+w_{101} &= w_{010}+w_{110}. \label{eq:hc-eq-hb}
\end{align}

Assume $w_{101}>0$. If $w_{010}>0$, we are done. If $w_{010}=0$, then
\begin{align*}
    w_{011}=w_{100}+w_{101}>0,
    \qquad
    w_{110}=w_{001}+w_{101}>0
\end{align*}
by \eqref{eq:ha-eq-hb} and \eqref{eq:hc-eq-hb}. This gives the claimed alternative.
\end{proof}

\subsection{Perturbation lemmas}
\label{sec:perturbation-lemmas}
The next step is to show that each forbidden local configuration makes the induced point in the joint range non-extreme. The general strategy is to build small feasible perturbations of the staircase channel that preserve row-stochasticity and $\varepsilon$-LDP, and then choose them so that the original point becomes a nontrivial midpoint of two distinct feasible image points.

We first record a simple perturbation criterion for staircase channels.
\begin{lemma}[Perturbation criterion for staircase channels]
\label{lemma:perturbation-criterion-staircase}
Let $Q\in\mathbb{R}_+^{k\times m}$ be a row-stochastic staircase $\varepsilon$-LDP channel with $\varepsilon>0$. Let $D\in\mathbb{R}^{k\times m}$ be supported on active columns of $Q$, and assume
\begin{align*}
    D\mathbf{1}=0.
\end{align*}
Assume also that for every modified column $y$, every row $h$ that is high in column $y$, and every row $\ell$ that is low in column $y$,
\begin{align}
    \label{eq:perturbation-high-low-inequality}
    D_{h,y}\le e^\varepsilon D_{\ell,y}.
\end{align}
Then $Q+\gamma D$ is row-stochastic and $\varepsilon$-LDP for all sufficiently small $\gamma>0$.
\end{lemma}

\begin{proof}
Since $D\mathbf{1}=0$, every row sum of $Q+\gamma D$ equals the corresponding row sum of $Q$, hence equals $1$.

Because $D$ is supported on active columns and every active staircase column is strictly positive entrywise, each modified entry of $Q$ starts strictly positive. Since there are finitely many modified entries, there exists $\gamma_0>0$ such that $Q+\gamma D\ge 0$ for all $0<\gamma<\gamma_0$.

Fix a modified column $y$. In a staircase column, every entry is either $\theta_y$ or $e^\varepsilon\theta_y$ for some $\theta_y>0$. Thus:
\begin{itemize}
    \item if both rows are high, or both are low, then the two entries are equal, so the ordered privacy inequality is
    \begin{align*}
        Q_{i,y}\le e^\varepsilon Q_{i',y},
    \end{align*}
    which is strict because $\varepsilon>0$;
    \item if row $h$ is high and row $\ell$ is low, then
    \begin{align*}
        Q_{h,y}=e^\varepsilon Q_{\ell,y},
    \end{align*}
    so this is the only potentially tight ordered privacy inequality;
    \item if row $\ell$ is low and row $h$ is high, then
    \begin{align*}
        Q_{\ell,y}<e^\varepsilon Q_{h,y},
    \end{align*}
    so that ordered inequality is also strict.
    \end{itemize}
    Therefore it is enough to check the high-to-low pairs. For such a pair,
    \begin{align*}
    Q_{h,y}=e^\varepsilon Q_{\ell,y},
\end{align*}
so after perturbation we get
\begin{align*}
    Q_{h,y}+\gamma D_{h,y}\le e^\varepsilon\bigl(Q_{\ell,y}+\gamma D_{\ell,y}\bigr)
\end{align*}
by \eqref{eq:perturbation-high-low-inequality}.

All other ordered privacy inequalities are strict at $\gamma=0$, so by continuity they remain valid for all sufficiently small $\gamma>0$. Since there are only finitely many such inequalities, shrinking $\gamma_0$ if necessary gives a single $\gamma_0>0$ for which every privacy inequality holds for all $0<\gamma<\gamma_0$. Therefore $Q+\gamma D$ is row-stochastic and $\varepsilon$-LDP for all sufficiently small $\gamma>0$.
\end{proof}

The next lemma adapts the midpoint perturbation argument used in the proof of Theorem~6 of \cite{Pensia2025BinaryTesting}. There, Pensia et al.\ rule out a non-threshold deterministic channel by showing that the induced point in the fixed-output joint range is not extreme. We use the same perturbative idea here for staircase $\varepsilon$-LDP channels containing a zigzag configuration.

\begin{lemma}[Zigzag configuration implies non-extremality in the joint range]
\label{lemma:zigzag-implies-non-extremality}
Let $p,q \in \Delta_k$ be in reduced strict extended likelihood-ratio order: all input symbols with $p_i=q_i=0$ have been deleted, equal extended likelihood ratios have been merged, and
\begin{align*}
    0\le r_1<r_2<\cdots<r_k\le \infty,
    \qquad
    r_i:=\frac{p_i}{q_i},
\end{align*}
with the convention that $p_i/0=\infty$ when $p_i>0$.
Let $Q \in \mathbb{R}_+^{k\times m}$ be a row-stochastic staircase $\varepsilon$-LDP channel with $\varepsilon>0$. If $Q$ contains a zigzag configuration, then the induced point $(Q^\top p,Q^\top q)$ is not an extreme point of
\begin{align*}
    \mathcal{R}_m(p,q)
    :=
    \{(T^\top p,T^\top q):
    T\in\mathbb{R}_+^{k\times m}
    \text{ is row-stochastic and }\varepsilon\text{-LDP}\}.
\end{align*}
Equivalently, there exist row-stochastic $\varepsilon$-LDP channels $Q',Q''\in\mathbb{R}_+^{k\times m}$ such that
\begin{align*}
    (Q^\top p,Q^\top q)
    =
    \tfrac12(Q'^\top p,Q'^\top q)
    +
    \tfrac12(Q''^\top p,Q''^\top q),
\end{align*}
and
\begin{align*}
    (Q'^\top p,Q'^\top q)\neq (Q''^\top p,Q''^\top q).
\end{align*}
Moreover, the channels $Q'$ and $Q''$ may be chosen so that they differ from $Q$ only in the two columns participating in the zigzag configuration.
\end{lemma}
\input{appendix/A_extreme_refinements/perturbation_zigzag_proof}

We now turn to the $3$-cycle configuration.

\begin{lemma}[$3$-cycle configuration implies non-extremality in the joint range]
\label{lemma:3-cycle-implies-non-extremality-in-joint-range}
Let $p,q \in \Delta_k$ and let $r_i := \frac{p_i}{q_i}$ be strictly increasing in the LR order. Let $Q \in \mathbb{R}_+^{k\times m}$ be a row-stochastic staircase $\varepsilon$-LDP channel with $\varepsilon > 0$. Suppose $Q$ contains a $3$-cycle configuration. Then the induced point $(Q^\top p, Q^\top q)$ is not an extreme point of the joint range
\begin{align*}
    \mathcal{R}_m(p,q):=\{(T^\top p,T^\top q): T \in \mathbb{R}_+^{k\times m} \text{ is }\varepsilon\text{-LDP and row-stochastic}\}.
\end{align*}
Equivalently, there exist row-stochastic $\varepsilon$-LDP channels $Q^+,Q^-\in\mathbb{R}_+^{k\times m}$ such that
\begin{align*}
    (Q^\top p,Q^\top q)
    =
    \tfrac12\bigl((Q^+)^\top p,(Q^+)^\top q\bigr)
    +
    \tfrac12\bigl((Q^-)^\top p,(Q^-)^\top q\bigr),
\end{align*}
and
\begin{align*}
    \bigl((Q^+)^\top p,(Q^+)^\top q\bigr)\neq \bigl((Q^-)^\top p,(Q^-)^\top q\bigr).
\end{align*}
Moreover, the channels $Q^+$ and $Q^-$ may be chosen so that they differ from $Q$ only in columns participating in the $3$-cycle configuration.
\end{lemma}
\input{appendix/A_extreme_refinements/perturbation_3_cycle_proof}

\subsection{Proof of Proposition~\ref{prop:extreme-refinement}}
\label{sec:proof-extreme-refinement}

\begin{proof}[Proof of Proposition~\ref{prop:extreme-refinement}]
Let $(u,v)\in \mathcal R_\ell(p,q)$ be an extreme point, and let $Q$ be an $\ell$-output row-stochastic $\varepsilon$-LDP channel such that
\begin{align*}
    (Q^\top p,Q^\top q)=(u,v).
\end{align*}

By Lemma~\ref{lem:nested-staircase-refinement}, there exist a staircase $\varepsilon$-LDP channel $\widetilde Q$ and a deterministic map $\tau:\widetilde{\mathcal Y}\to[\ell]$ such that
\begin{align*}
    Q(i,y)=\sum_{\tilde y:\tau(\tilde y)=y}\widetilde Q(i,\tilde y)
    \qquad (i\in[k],\ y\in[\ell]),
\end{align*}
and, for each $y\in[\ell]$, the high sets of the fine columns merged into $y$ are nested.

Suppose some active column of $\widetilde Q$ is not LR-contiguous. Then for some LR-ordered triple $a<b<c$ it has pattern $(101)$, so $w_{101}>0$ for the weights defined in Lemma~\ref{lemma:any-101-triple-forces-zigzag-or-3-cycle}. By that lemma, either $w_{010}>0$ or $w_{110}>0$ and $w_{011}>0$. Hence $\widetilde Q$ contains either a zigzag configuration or a $3$-cycle configuration on that triple. These columns are distinct, since one column cannot realize two different patterns on the same triple. Let $S$ be the set of columns involved in that local configuration.

No two columns in $S$ can be merged into the same coarse output. Indeed, for a fixed coarse output, the high sets of the fine columns merged into it are nested; but in a zigzag the two patterns $(101)$ and $(010)$ are not nested on $\{a,b,c\}$, and in a $3$-cycle the three patterns $(101)$, $(110)$, and $(011)$ are pairwise non-nested on $\{a,b,c\}$.

By the corresponding perturbation lemma (Lemma~\ref{lemma:zigzag-implies-non-extremality} in the zigzag case and Lemma~\ref{lemma:3-cycle-implies-non-extremality-in-joint-range} in the $3$-cycle case) there exist feasible row-stochastic $\varepsilon$-LDP channels $\widetilde Q^+$ and $\widetilde Q^-$, differing from $\widetilde Q$ only in columns from $S$, such that
\begin{align*}
    (\widetilde Q^\top p,\widetilde Q^\top q)
    =
    \tfrac12\bigl((\widetilde Q^+)^\top p,(\widetilde Q^+)^\top q\bigr)
    +
    \tfrac12\bigl((\widetilde Q^-)^\top p,(\widetilde Q^-)^\top q\bigr),
\end{align*}
with distinct endpoints.

Let $\Pi$ be the $|\widetilde{\mathcal Y}|\times \ell$ post-processing matrix induced by $\tau$, and define
\begin{align*}
    Q^+ &:= \widetilde Q^+\Pi, & Q^- &:= \widetilde Q^-\Pi.
\end{align*}
Since deterministic post-processing preserves row-stochasticity and $\varepsilon$-LDP, both $Q^+$ and $Q^-$ are feasible $\ell$-output channels. Moreover, by linearity of post-processing,
\begin{align*}
    (Q^\top p,Q^\top q) = \tfrac12\bigl((Q^+)^\top p,(Q^+)^\top q\bigr) + \tfrac12\bigl((Q^-)^\top p,(Q^-)^\top q\bigr).
\end{align*}

It remains to show that the two coarse endpoints are still distinct. Let
\begin{align*}
    \Delta_p &:= (\widetilde Q^+)^\top p-(\widetilde Q^-)^\top p, & \Delta_q &:= (\widetilde Q^+)^\top q-(\widetilde Q^-)^\top q.
\end{align*}
Since $\widetilde Q^+$ and $\widetilde Q^-$ differ only in columns from $S$, both $\Delta_p$ and $\Delta_q$ are supported on $S$. They are not both zero because the fine endpoints are distinct.

If the coarse endpoints were equal, then
\begin{align*}
    (Q^+)^\top p &= (Q^-)^\top p, & (Q^+)^\top q &= (Q^-)^\top q.
\end{align*}
Since $Q^+=\widetilde Q^+\Pi$ and $Q^-=\widetilde Q^-\Pi$, this would give
\begin{align*}
    \Pi^\top \Delta_p &= 0, & \Pi^\top \Delta_q &= 0.
\end{align*}
But the columns in $S$ are sent to distinct coarse outputs, so the restriction of $\Pi^\top$ to vectors supported on $S$ is injective. Hence $\Delta_p=\Delta_q=0$, a contradiction. Therefore
\begin{align*}
    \bigl((Q^+)^\top p,(Q^+)^\top q\bigr)\neq \bigl((Q^-)^\top p,(Q^-)^\top q\bigr).
\end{align*}

Thus $(u,v)=(Q^\top p,Q^\top q)$ is a nontrivial midpoint of two distinct points in $\mathcal R_\ell(p,q)$, contradicting the assumption that $(u,v)$ is extreme.

We conclude that every active column of $\widetilde Q$ is LR-contiguous. Since $(u,v)$ is realized by a deterministic post-processing of $\widetilde Q$, this proves the proposition.
\end{proof}

%% file: figures/proof_logic_figure.tex
\begin{tikzpicture}[
    node distance=15mm,
    box/.style={rectangle, rounded corners=3pt, draw=black!55,
                fill=blockblue!35, align=center, minimum height=13mm,
                minimum width=82mm, font=\small, inner sep=6pt},
    sbox/.style={rectangle, rounded corners=3pt, draw=black!55,
                 fill=blockyellow!55, align=center, minimum height=16mm,
                 minimum width=54mm, font=\small\bfseries, inner sep=6pt},
    redbox/.style={box, fill=red!10, draw=red!50, font=\small\itshape},
    gbox/.style={box, fill=blockgreen!55, draw=black!55,
                 font=\small\bfseries},
    arr/.style={-Latex, line width=1pt, color=black!75},
    darr/.style={arr, dashed, color=red!55!black},
    arrlbl/.style={font=\small, fill=white, inner sep=2pt,
                   align=center}
]

\node[box] (start)
  {Consider $(u,v)\in\mathcal R_\ell(p,q)$ extreme,\\
   realized by some $Q\in\mathcal Q_{\eps,\ell}$};

\node[box, below=of start] (staircase)
  {Staircase refinement $\widetilde Q$, post-processed back to $Q$\\
   by a deterministic $\tau$, with nested per-output high sets};
\draw[arr] (start) -- (staircase)
  node[midway, right=2pt, arrlbl] {Lemma~\ref{lem:nested-staircase-refinement}};

\node[redbox, below=of staircase] (triple)
  {\textbf{Suppose} some active column of $\widetilde Q$ has\\
   non-contiguous high set: pattern $(1,0,1)$ on a triple $a<b<c$};
\draw[darr] (staircase) -- (triple)
  node[midway, right=2pt, arrlbl] {for contradiction};

\node[sbox] (zigzag) at ($(triple.south)+(-3.0,-1.9)$)
  {Zigzag\\ \small\mdseries $(1,0,1)+(0,1,0)$};
\node[sbox] (cycle) at ($(triple.south)+(3.0,-1.9)$)
  {$3$-cycle\\ \small\mdseries $(1,0,1)+(1,1,0)+(0,1,1)$};
\draw[arr] (triple.south) -- (zigzag.north);
\draw[arr] (triple.south) -- (cycle.north);
\node[arrlbl] at ($(triple.south)+(0,-0.65)$) {Lemma~\ref{lemma:any-101-triple-forces-zigzag-or-3-cycle}};

\node[box] (perturb)
  at ($(zigzag.south)!0.5!(cycle.south)+(0,-2.2)$)
  {Construct feasible $\widetilde Q^+,\widetilde Q^-$, differing from $\widetilde Q$ only\\
   on the involved columns, with distinct image pair\\
   whose midpoint is $\widetilde Q$'s image pair};
\draw[arr] (zigzag.south) -- (zigzag.south |- perturb.north)
  node[pos=0.5, left=1pt, arrlbl] {Lemma~\ref{lemma:zigzag-implies-non-extremality}};
\draw[arr] (cycle.south) -- (cycle.south |- perturb.north)
  node[pos=0.5, right=1pt, arrlbl] {Lemma~\ref{lemma:3-cycle-implies-non-extremality-in-joint-range}};

\node[box, below=of perturb] (push)
  {Post-processing $\widetilde Q^\pm$ by $\tau$ yields feasible $\eps$-LDP $Q^\pm$\\
   with distinct image pair whose midpoint is still $(u,v)$};
\draw[arr] (perturb) -- (push)
  node[midway, right=2pt, arrlbl] {apply $\tau$};

\node[redbox, below=of push] (contra)
  {$(u,v)$ is a nontrivial midpoint of two points in $\mathcal R_\ell(p,q)$};
\draw[arr] (push) -- (contra);

\node[gbox, below=of contra] (concl)
  {Every active column of $\widetilde Q$ is LR-contiguous};
\draw[arr, line width=1.05pt] (contra) -- (concl)
  node[midway, right=2pt, arrlbl] {contradicts\\extremality};

\end{tikzpicture}

%% file: appendix/A_extreme_refinements/perturbation_zigzag_proof.tex
\begin{proof}
Let $a<b<c$ and $j\neq j'$ witness the zigzag configuration, so that on rows $(a,b,c)$ the columns $j$ and $j'$ have patterns
\begin{align*}
    101,\qquad 010.
\end{align*}
Let $e_j,e_{j'}$ be the standard basis vectors in $\mathbb{R}^m$.

Since the likelihood ratios are strictly ordered and $a<b<c$, the rows $a$ and $b$ have finite likelihood ratio, hence $q_a,q_b>0$. The row $c$ may have finite likelihood ratio, in which case $q_c>0$, or it may be the unique infinite-likelihood-ratio endpoint, in which case $q_c=0$ and $p_c>0$.

Define $A,B\in\mathbb{R}^{k\times m}$ by
\begin{align*}
    A_{a,j}=-1,\qquad A_{a,j'}=1,\qquad
    A_{b,j}=\frac{q_a}{q_b},\qquad A_{b,j'}=-\frac{q_a}{q_b},
\end{align*}
\begin{align*}
    B_{b,j}=q_c,\qquad B_{b,j'}=-q_c,\qquad
    B_{c,j}=-q_b,\qquad B_{c,j'}=q_b,
\end{align*}
and all other entries equal to $0$.

Equivalently,
\begin{align*}
    A=
    \begin{bmatrix}
        \cdots & \cdots & \cdots & \cdots & \cdots \\
        \cdots & -1 & \cdots & +1 & \cdots \\
        \cdots & \cdots & \cdots & \cdots & \cdots \\
        \cdots & +\dfrac{q_a}{q_b} & \cdots & -\dfrac{q_a}{q_b} & \cdots \\
        \cdots & \cdots & \cdots & \cdots & \cdots \\
        \cdots & 0 & \cdots & 0 & \cdots \\
        \cdots & \cdots & \cdots & \cdots & \cdots
    \end{bmatrix},
    \qquad
    B=
    \begin{bmatrix}
        \cdots & \cdots & \cdots & \cdots & \cdots \\
        \cdots & 0 & \cdots & 0 & \cdots \\
        \cdots & \cdots & \cdots & \cdots & \cdots \\
        \cdots & +q_c & \cdots & -q_c & \cdots \\
        \cdots & \cdots & \cdots & \cdots & \cdots \\
        \cdots & -q_b & \cdots & +q_b & \cdots \\
        \cdots & \cdots & \cdots & \cdots & \cdots
    \end{bmatrix}.
\end{align*}
When $q_c>0$, this $B$ is exactly $q_c$ times the original finite-ratio perturbation with entries $1,-1,-q_b/q_c,q_b/q_c$. When $q_c=0$, it reduces to the endpoint perturbation supported only on row $c$.

Both $A$ and $B$ are supported on the active columns $j,j'$ and satisfy
\begin{align*}
    A\mathbf{1}=\mathbf{0},\qquad B\mathbf{1}=\mathbf{0}.
\end{align*}
Moreover, in each modified column, every high-row perturbation is nonpositive and every low-row perturbation is nonnegative: for the rows not listed above, the perturbation is $0$. Hence the pairwise inequalities of Lemma~\ref{lemma:perturbation-criterion-staircase} hold for both $A$ and $B$. It follows that
\begin{align*}
    Q+\gamma A,\qquad Q+\gamma B
\end{align*}
are row-stochastic and $\varepsilon$-LDP for all sufficiently small $\gamma>0$.

Their images are
\begin{align*}
    A^\top q=\mathbf{0},
    \qquad
    A^\top p
    =
    \left(-p_a+\frac{q_a}{q_b}p_b\right)(e_j-e_{j'})
    =
    q_a(r_b-r_a)(e_j-e_{j'}),
\end{align*}
and
\begin{align*}
    B^\top q=\mathbf{0},
    \qquad
    B^\top p
    =
    (q_c p_b-q_b p_c)(e_j-e_{j'}).
\end{align*}
Define
\begin{align*}
    \Delta_{ab}:=q_a(r_b-r_a)>0,
    \qquad
    \Delta_{bc}:=q_b p_c-q_c p_b.
\end{align*}
Then
\begin{align*}
    A^\top p=\Delta_{ab}(e_j-e_{j'}),
    \qquad
    B^\top p=-\Delta_{bc}(e_j-e_{j'}).
\end{align*}
We have $\Delta_{bc}>0$. Indeed, if $q_c>0$, then
\begin{align*}
    \Delta_{bc}
    =
    q_bq_c(r_c-r_b)>0,
\end{align*}
while if $q_c=0$, then $p_c>0$ and
\begin{align*}
    \Delta_{bc}=q_b p_c>0.
\end{align*}
Set
\begin{align*}
    \lambda:=\frac{\Delta_{ab}}{\Delta_{bc}}
    =
    \frac{q_a(r_b-r_a)}{q_b p_c-q_c p_b}.
\end{align*}
Then $\lambda>0$, and
\begin{align*}
    (\lambda B)^\top q=\mathbf{0},
    \qquad
    (\lambda B)^\top p=-A^\top p.
\end{align*}
Since $\lambda\ge 0$, the same sign condition also holds for $\lambda B$, so
\begin{align*}
    Q+\gamma A,\qquad Q+\gamma \lambda B
\end{align*}
are row-stochastic and $\varepsilon$-LDP for all sufficiently small $\gamma>0$.

Choose $\gamma>0$ small enough that both
\begin{align*}
    Q':=Q+\gamma A,
    \qquad
    Q'':=Q+\gamma \lambda B
\end{align*}
are row-stochastic and $\varepsilon$-LDP. Then
\begin{align*}
    (Q^\top p,Q^\top q)
    =
    \frac12(Q'^\top p,Q'^\top q)
    +
    \frac12(Q''^\top p,Q''^\top q),
\end{align*}
because
\begin{align*}
    (A^\top p,A^\top q)+((\lambda B)^\top p,(\lambda B)^\top q)=0.
\end{align*}
This convex combination is nontrivial since
\begin{align*}
    Q'^\top p-Q^\top p
    =
    \gamma A^\top p
    =
    \gamma q_a(r_b-r_a)(e_j-e_{j'})\neq 0.
\end{align*}
Hence $(Q^\top p,Q^\top q)$ is not an extreme point of $\mathcal{R}_m(p,q)$.

By construction, $Q'$ and $Q''$ differ from $Q$ only in the two columns $j$ and $j'$.
\end{proof}

\paragraph{In matrix form.}
The perturbed channels $Q'=Q+\gamma A$ and $Q''=Q+\gamma\lambda B$ have entries
\begin{align*}
    Q'=
    \begin{bmatrix}
        \cdots & \cdots & \cdots & \cdots & \cdots \\
        \cdots & e^\varepsilon \theta_j-\gamma & \cdots & \theta_{j'}+\gamma & \cdots \\
        \cdots & \cdots & \cdots & \cdots & \cdots \\
        \cdots & \theta_j+\gamma \frac{q_a}{q_b} & \cdots & e^\varepsilon \theta_{j'}-\gamma \frac{q_a}{q_b} & \cdots \\
        \cdots & \cdots & \cdots & \cdots & \cdots \\
        \cdots & e^\varepsilon \theta_j & \cdots & \theta_{j'} & \cdots \\
        \cdots & \cdots & \cdots & \cdots & \cdots
    \end{bmatrix},
    \qquad
    Q''=
    \begin{bmatrix}
        \cdots & \cdots & \cdots & \cdots & \cdots \\
        \cdots & e^\varepsilon \theta_j & \cdots & \theta_{j'} & \cdots \\
        \cdots & \cdots & \cdots & \cdots & \cdots \\
        \cdots & \theta_j+\gamma\lambda q_c & \cdots & e^\varepsilon \theta_{j'}-\gamma\lambda q_c & \cdots \\
        \cdots & \cdots & \cdots & \cdots & \cdots \\
        \cdots & e^\varepsilon \theta_j-\gamma\lambda q_b & \cdots & \theta_{j'}+\gamma\lambda q_b & \cdots \\
        \cdots & \cdots & \cdots & \cdots & \cdots
    \end{bmatrix},
\end{align*}
with $\lambda = \frac{q_a(r_b-r_a)}{q_b p_c-q_c p_b}$ and remaining entries (denoted $\cdots$) unchanged from $Q$.

%% file: appendix/A_extreme_refinements/perturbation_3_cycle_proof.tex
\begin{proof}
Let $e_1,\dots,e_m$ be the standard basis of $\mathbb{R}^m$. Choose rows $a<b<c$ and distinct columns $u,v,w$ witnessing the $3$-cycle, so that on rows $(a,b,c)$ the columns $u,v,w$ have patterns
\begin{align*}
    101,\qquad 110,\qquad 011.
\end{align*}
We take column $u$ as the reference column. For each row $i$, define
\begin{align*}
    \alpha_i=
    \begin{cases}
        1, & \text{if row $i$ is high in column $u$,}\\
        e^{-\varepsilon}, & \text{if row $i$ is low in column $u$,}
    \end{cases}
\end{align*}
and set
\begin{align*}
    \bar q=\sum_{i=1}^k q_i\alpha_i,\qquad
    \bar p=\sum_{i=1}^k p_i\alpha_i.
\end{align*}
Then $\bar q>0$.

We distinguish two cases:
\begin{align*}
    \text{Case I: }\bar p\ge r_b\bar q,
    \qquad
    \text{Case II: }\bar p<r_b\bar q.
\end{align*}
In the two cases, define $x,z,\Delta,\lambda$ as follows:
\begin{align*}
    \renewcommand{\arraystretch}{2.0}
    \begin{array}{c|c|c|c|c}
         & x & z & \Delta & \lambda \\ \hline
        \text{Case I} & a & w & q_a(r_b-r_a) &
        \dfrac{\bar p-r_b\bar q}{\Delta} \\
        \text{Case II} & c & v & p_c-q_c r_b &
        \dfrac{r_b\bar q-\bar p}{\Delta}.
    \end{array}
\end{align*}
In both cases, row $x$ is high in column $u$ and low in column $z$, while row $b$ is low in column $u$ and high in column $z$. Also $\Delta>0$ in both cases. In Case I this follows from $q_a>0$ and $r_b>r_a$. In Case II, if $q_c>0$ then
\begin{align*}
    \Delta=p_c-q_c r_b=q_c(r_c-r_b)>0,
\end{align*}
while if $q_c=0$, then by the reduced extended likelihood-ratio order $p_c>0$, and hence $\Delta=p_c>0$. Therefore $\lambda\ge 0$ in both cases.

Define $D_+$, $G$, and $C$ by
\begin{align*}
    (D_+)_{b,u}=\frac{\bar q}{q_b},
    \qquad
    (D_+)_{b,z}=-\frac{\bar q}{q_b},
\end{align*}
\begin{align*}
    G_{i,u}=-\alpha_i,\qquad G_{i,z}=+\alpha_i \quad (i\in[k]),
\end{align*}
and
\begin{align*}
    C_{x,u}=-1,\qquad C_{x,z}=1,\qquad
    C_{b,u}=\frac{q_x}{q_b},\qquad C_{b,z}=-\frac{q_x}{q_b},
\end{align*}
with all other entries equal to zero. Finally set
\begin{align*}
    D_-:=G+\lambda C.
\end{align*}

Each of $D_+$, $G$, and $C$ is supported on the two columns $u,z$ and has row sums equal to zero. Hence $D_-\mathbf 1=0$ as well.

We now verify the high--low inequality in Lemma~\ref{lemma:perturbation-criterion-staircase}. For $D_+$, in column $u$ the only nonzero perturbation is on the low row $b$ and is nonnegative, while in column $z$ the only nonzero perturbation is on the high row $b$ and is nonpositive. Thus $D_+$ satisfies the required high--low inequality. For $C$, in column $u$ the nonzero perturbations are $-1$ on the high row $x$ and $q_x/q_b\ge 0$ on the low row $b$, while in column $z$ they are $1$ on the low row $x$ and $-q_x/q_b\le 0$ on the high row $b$. Thus $C$ also satisfies the high--low inequality. For $G$, in column $u$ every high-row perturbation is $-1$ and every low-row perturbation is $-e^{-\varepsilon}$, so the high--low inequality holds with equality. In column $z$, all entries of $G$ lie in $\{1,e^{-\varepsilon}\}$, and hence for every high row $h$ and low row $\ell$,
\begin{align*}
    G_{h,z}\le 1 \le e^\varepsilon G_{\ell,z}.
\end{align*}
Since $\lambda\ge 0$, the same high--low inequality holds for $D_-=G+\lambda C$. Therefore, by Lemma~\ref{lemma:perturbation-criterion-staircase},
\begin{align*}
    Q+\gamma D_+,\qquad Q+\gamma D_-
\end{align*}
are row-stochastic and $\varepsilon$-LDP for all sufficiently small $\gamma>0$.

We now compute the induced image increments. By construction,
\begin{align*}
    D_+^\top q=\bar q(e_u-e_z),
    \qquad
    D_+^\top p=r_b\bar q(e_u-e_z),
\end{align*}
and
\begin{align*}
    G^\top q=-\bar q(e_u-e_z),
    \qquad
    G^\top p=-\bar p(e_u-e_z).
\end{align*}
Moreover,
\begin{align*}
    C^\top q=
    \Bigl(-q_x+\frac{q_x}{q_b}q_b\Bigr)e_u
    +
    \Bigl(q_x-\frac{q_x}{q_b}q_b\Bigr)e_z
    =\mathbf 0,
\end{align*}
and
\begin{align*}
    C^\top p=
    \Bigl(-p_x+\frac{q_x}{q_b}p_b\Bigr)e_u
    +
    \Bigl(p_x-\frac{q_x}{q_b}p_b\Bigr)e_z
    =
    (q_xr_b-p_x)(e_u-e_z).
\end{align*}
By the definitions of $\Delta$ and $\lambda$, in both cases
\begin{align*}
    \lambda(q_xr_b-p_x)=\bar p-r_b\bar q.
\end{align*}
Indeed, in Case I, $q_xr_b-p_x=q_a(r_b-r_a)=\Delta$, while in Case II, $q_xr_b-p_x=q_c r_b-p_c=-\Delta$.

Therefore,
\begin{align*}
    D_-^\top q=G^\top q+\lambda C^\top q=-\bar q(e_u-e_z),
\end{align*}
and
\begin{align*}
    D_-^\top p
    =
    G^\top p+\lambda C^\top p
    =
    \bigl(-\bar p+\lambda(q_xr_b-p_x)\bigr)(e_u-e_z)
    =
    -r_b\bar q(e_u-e_z).
\end{align*}
Hence
\begin{align*}
    (D_-^\top p,D_-^\top q)=-(D_+^\top p,D_+^\top q).
\end{align*}

Choose $\gamma>0$ small enough that both
\begin{align*}
    Q^+:=Q+\gamma D_+,
    \qquad
    Q^-:=Q+\gamma D_-
\end{align*}
are row-stochastic and $\varepsilon$-LDP. Then, by linearity,
\begin{align*}
    (Q^\top p,Q^\top q)
    =
    \frac12\bigl((Q^+)^\top p,(Q^+)^\top q\bigr)
    +
    \frac12\bigl((Q^-)^\top p,(Q^-)^\top q\bigr).
\end{align*}
This convex combination is nontrivial because
\begin{align*}
    (Q^+)^\top q-(Q^-)^\top q
    =
    2\gamma\bar q(e_u-e_z)\neq 0,
\end{align*}
since $\bar q>0$ and $u\neq z$. Hence
\begin{align*}
    \bigl((Q^+)^\top p,(Q^+)^\top q\bigr)
    \neq
    \bigl((Q^-)^\top p,(Q^-)^\top q\bigr),
\end{align*}
so $(Q^\top p,Q^\top q)$ is not an extreme point of $\mathcal R_m(p,q)$.

By construction, $Q^+$ and $Q^-$ differ from $Q$ only in the two modified columns, namely $\{u,w\}$ in Case I and $\{u,v\}$ in Case II. In particular, they differ from $Q$ only in columns participating in the $3$-cycle configuration.
\end{proof}

%% file: appendix/B_convex_hull_spr.tex
\section{Convex hull of SPR mechanisms}
\label{sec:proof-convex-hull-partition-rr}

We isolate the combinatorial core of the proof as a decomposition lemma for weighted interval families with constant coverage.

\begin{lemma}[Layer decomposition of constant-coverage interval families]
\label{lem:constant-coverage-interval-decomposition}
Suppose we have finitely many nonempty intervals $I_y\subseteq [k]$, indexed by $y\in A$, and positive weights $\theta_y>0$. Assume that for each $i\in [k]$, the total weight of intervals covering $i$,
\begin{align*}
    h(i):=\sum_{y:\, i\in I_y}\theta_y,
\end{align*}
is the same positive constant. Then there exist subsets $\Pi_1,\dots,\Pi_N\subseteq A$ and numbers $\lambda_1,\dots,\lambda_N>0$ such that:
\begin{enumerate}
\item for each $t$, the intervals $\{I_y:y\in \Pi_t\}$ form a partition of $[k]$;
\item for each $y\in A$,
\begin{align*}
    \theta_y=\sum_{t:\, y\in \Pi_t}\lambda_t.
\end{align*}
\end{enumerate}
\end{lemma}
The lemma says that if weighted intervals overlap in such a way that every position $i \in [k]$ sees the same total weight, then the whole system can be peeled into layers, each of which is a partition of $[k]$. The coefficients $\lambda_t$ record how much weight each layer contributes, and adding the layers back together recovers the original interval weights. This can be viewed as the standard flow-decomposition theorem \cite{Ahuja1993NetworkFlows} applied to the DAG that identifies each interval $[a:b]$ with an edge $a-1\to b$.
\begin{proof}
We peel off the interval family one partition layer at a time. At each stage, we keep track of the remaining weight on each interval and choose a collection of intervals that partitions $[k]$; we then subtract the same amount from every interval in that layer.

Formally, start with remaining weights $w_y:=\theta_y$. At any stage, let
\begin{align*}
    h_{\mathrm{cur}}(i):=\sum_{y:\, i\in I_y} w_y
\end{align*}
be the total remaining weight covering $i$. Initially $h_{\mathrm{cur}}$ is constant by assumption, and after each layer is removed it will remain constant.

As long as some $w_y$ is positive, we build one layer from left to right. Since some positive-weight interval is nonempty, there is at least one point with positive remaining coverage. Because $h_{\mathrm{cur}}$ is constant, this implies that $h_{\mathrm{cur}}(i)>0$ for every $i\in [k]$, in particular for $i=1$. Hence there exists a positive-weight interval containing $1$, and since it is an interval, it must be of the form $[1:b_1]$ for some $b_1$.
Now suppose we have already chosen disjoint positive-weight intervals
\begin{align*}
    [1:b_1],\ [b_1+1:b_2],\ \dots,\ [b_{r-1}+1:b_r]
\end{align*}
and that $b_r<k$. Because $h_{\mathrm{cur}}(b_r)=h_{\mathrm{cur}}(b_r+1)$, the total remaining weight of intervals ending at $b_r$ equals the total remaining weight of intervals starting at $b_r+1$:
\begin{align*}
    \sum_{y:\, b_r\in I_y,\ b_r+1\notin I_y} w_y
    =
    \sum_{y:\, b_r\notin I_y,\ b_r+1\in I_y} w_y.
\end{align*}
The last chosen interval ends at $b_r$, so the left-hand side is positive. Hence the right-hand side is also positive, and there exists a positive-weight interval starting at $b_r+1$. Choose one as the next block.

This shows that the construction cannot stop before reaching $k$. Therefore it produces a set of intervals $\{I_y:y\in \Pi_t\}$ that partitions $[k]$.

Now let
\begin{align*}
    \lambda_t:=\min_{y\in \Pi_t} w_y>0,
\end{align*}
and subtract $\lambda_t$ from the weight of every interval in this layer:
\begin{align*}
    w_y\leftarrow
    \begin{cases}
    w_y-\lambda_t, & y\in \Pi_t,\\
    w_y, & y\notin \Pi_t.
    \end{cases}
\end{align*}
Since the intervals in $\Pi_t$ partition $[k]$, each point of $[k]$ lies in exactly one of them, so $h_{\mathrm{cur}}(i)$ decreases by exactly $\lambda_t$ for every $i$. Thus $h_{\mathrm{cur}}$ remains constant. Also, by the choice of $\lambda_t$, at least one interval in $\Pi_t$ now has remaining weight zero. Therefore the process terminates after finitely many steps.

Finally, each time $I_y$ appears in a layer $\Pi_t$, we subtract exactly $\lambda_t$ from its remaining weight. Since $I_y$ starts with weight $\theta_y$ and ends with weight $0$, the total amount removed from it is
\begin{align*}
    \theta_y=\sum_{t:\, y\in \Pi_t}\lambda_t.
\end{align*}
This proves the lemma.
\end{proof}

\begin{proof}[Proof of Proposition~\ref{prop:convex-hull-partition-rr}]
We first view $Q$ as a weighted family of LR-contiguous intervals, and then apply Lemma~\ref{lem:constant-coverage-interval-decomposition}. Example~\ref{ex:convex-decomposition-spr} at the end of this section illustrates the construction on a $5\times 5$ instance.

Let $A\subseteq[m]$ be the set of active columns of $Q$. For each $y\in A$, choose a scale $\theta_y>0$ and a nonempty LR-contiguous interval $I_y\subseteq [k]$ such that
\begin{align*}
    Q(i,y)=\theta_y\Bigl(1+(e^\eps-1)\mathbf 1\{i\in I_y\}\Bigr)
    \qquad (i\in [k]).
\end{align*}
(If column $y$ is constant, take $I_y=[k]$ and choose $\theta_y$ to be the common entry divided by $e^\eps$.) For $y\notin A$, we have $Q(i,y)=0$ for all $i$.

Define the ``coverage'' function
\begin{align*}
    h(i):=\sum_{y\in A:\, i\in I_y}\theta_y.
\end{align*}
Since $Q$ is row-stochastic,
\begin{align*}
    1=\sum_{y\in A}Q(i,y)
    =\sum_{y\in A}\theta_y+(e^\eps-1)h(i),
\end{align*}
so $h(i)$ is independent of $i$. It is also positive: because $A\neq\varnothing$ and each $I_y$ is nonempty, there exist $y\in A$ and $i\in I_y$, giving $h(i)\ge \theta_y>0$.

Lemma~\ref{lem:constant-coverage-interval-decomposition} therefore yields subsets $\Pi_1,\dots,\Pi_N\subseteq A$ and numbers $\lambda_1,\dots,\lambda_N>0$ such that, for each $t$, the family $\{I_y:y\in \Pi_t\}$ is a partition of $[k]$, and
\begin{align*}
    \theta_y=\sum_{t:\, y\in \Pi_t}\lambda_t \qquad (y\in A).
\end{align*}

For each layer $t$, the set $\Pi_t$ is a partition of $[k]$ into $s_t:=|\Pi_t|$ contiguous blocks. We now turn this partition into an $\eps$-LDP mechanism by applying $s_t$-ary randomized response to the block label. Thus, for row $i$, the unique block in $\Pi_t$ containing $i$ gets probability $e^\eps/(e^\eps+s_t-1)$, and each of the other $s_t-1$ active outputs gets probability $1/(e^\eps+s_t-1)$. We index these active outputs by the corresponding $y\in \Pi_t$, and keep the output alphabet $[m]$ by setting all other columns to zero. This gives the mechanism $\widehat Q^{(t)}$, which is just the SPR mechanism associated with $\Pi_t$, relabeled and zero-padded.

The layer decomposition tells us how much raw staircase weight comes from each layer: $\lambda_t$ is the thickness of layer $t$. Since $\widehat Q^{(t)}$ is the normalized SPR mechanism associated with that layer, we choose
\begin{align*}
    \beta_t:=\lambda_t(e^\eps+s_t-1)>0
\end{align*}
so that the normalization factor cancels. Indeed, for any $y\in \Pi_t$,
\begin{align*}
    \beta_t \widehat Q^{(t)}(i,y)
    = \lambda_t\Bigl(1+(e^\eps-1)\mathbf 1\{i\in I_y\}\Bigr),
\end{align*}
while for $y\notin \Pi_t$ we have $\widehat Q^{(t)}(i,y)=0$. Thus layer $t$ contributes exactly the raw staircase amount coming from the interval decomposition.

We now verify that
\begin{align*}
    Q=\sum_{t=1}^N \beta_t \widehat Q^{(t)}.
\end{align*}

If $y\notin A$, then both sides are identically zero in column $y$. If $y\in A$, then
\begin{align*}
\sum_{t=1}^N \beta_t \widehat Q^{(t)}(i,y)
&=
\sum_{t:\, y\in \Pi_t}
\lambda_t\Bigl(1+(e^\eps-1)\mathbf 1\{i\in I_y\}\Bigr)\\
&=
\Bigl(\sum_{t:\, y\in \Pi_t}\lambda_t\Bigr)
\Bigl(1+(e^\eps-1)\mathbf 1\{i\in I_y\}\Bigr)\\
&=
\theta_y\Bigl(1+(e^\eps-1)\mathbf 1\{i\in I_y\}\Bigr)\\
&=Q(i,y),
\end{align*}
using the decomposition of $\theta_y$. Hence the identity holds entrywise.

Both $Q$ and $\widehat Q^{(t)}$ are row-stochastic. Summing over the entries in any fixed row of $Q=\sum_{t=1}^N \beta_t \widehat Q^{(t)}$ therefore gives
\begin{align*}
    1=\sum_{t=1}^N \beta_t.
\end{align*}
Thus $Q$ is a convex combination of relabeled, zero-padded SPR mechanisms, as claimed.
\end{proof}
\begin{example}[Convex decomposition of a one-run staircase mechanism]
\label{ex:convex-decomposition-spr}
We consider a staircase mechanism $Q$ where each column has exactly one run of $e^\eps$, and exhibit a convex combination $Q=\beta_1\widehat Q^{(1)}+\beta_2\widehat Q^{(2)}+\beta_3\widehat Q^{(3)}$ with $\widehat Q^{(1)}, \widehat Q^{(2)}, \widehat Q^{(3)}$ all SPR. Take \(\varepsilon=1\), so \(e^\varepsilon=e\), and let $D:=17+14e$.
Consider the five output columns with high sets
\begin{align*}
    I_1=[1,1],\qquad
    I_2=[1,2],\qquad
    I_3=[2,2],\qquad
    I_4=[2,5],\qquad
    I_5=[3,5],
\end{align*}
and weights
\begin{align*}
    \theta_1=\frac{10}{D},\qquad
    \theta_2=\frac{4}{D},\qquad
    \theta_3=\frac{3}{D},\qquad
    \theta_4=\frac{7}{D},\qquad
    \theta_5=\frac{7}{D}.
\end{align*}
Define \(Q\in\mathbb R_+^{5\times 5}\) by
\begin{align*}
    Q(i,y)=\theta_y\Bigl(1+(e-1)\mathbf 1\{i\in I_y\}\Bigr).
\end{align*}
Equivalently,
\begin{align*}
    Q=
    \frac1D
    \begin{bmatrix}
    10e & 4e & 3 & 7 & 7\\
    10  & 4e & 3e& 7e& 7\\
    10  & 4  & 3 & 7e& 7e\\
    10  & 4  & 3 & 7e& 7e\\
    10  & 4  & 3 & 7e& 7e
    \end{bmatrix}.
\end{align*}
Each row sums to \(1\), since for example
\begin{align*}
    10e+4e+3+7+7=14e+17=D,
\end{align*}
and similarly for the other rows. The total weight of intervals containing each input is constant:
\begin{align*}
    \theta_1+\theta_2=\frac{14}{D},\qquad
    \theta_2+\theta_3+\theta_4=\frac{14}{D},\qquad
    \theta_4+\theta_5=\frac{14}{D}.
\end{align*}
Thus the weighted interval family has constant row-coverage. Now peel off partition layers:
\begin{itemize}
\item First layer:
\begin{align*}
    \Pi_1=\{I_1,I_4\}=\{[1,1],[2,5]\},
    \qquad
    \lambda_1=\min\{\theta_1,\theta_4\}=\frac7D.
\end{align*}
Subtracting \(\lambda_1\) leaves
\begin{align*}
    \theta_1'=\frac3D,\quad
    \theta_2'=\frac4D,\quad
    \theta_3'=\frac3D,\quad
    \theta_4'=0,\quad
    \theta_5'=\frac7D.
\end{align*}
\item Second layer:
\begin{align*}
    \Pi_2=\{I_2,I_5\}=\{[1,2],[3,5]\},
    \qquad
    \lambda_2=\min\{\theta_2',\theta_5'\}=\frac4D.
\end{align*}
Subtracting \(\lambda_2\) leaves
\begin{align*}
    \theta_1''=\frac3D,\quad
    \theta_2''=0,\quad
    \theta_3''=\frac3D,\quad
    \theta_4''=0,\quad
    \theta_5''=\frac3D.
\end{align*}
\item Third layer:
\begin{align*}
    \Pi_3=\{I_1,I_3,I_5\}=\{[1,1],[2,2],[3,5]\},
    \qquad
    \lambda_3=\frac3D.
\end{align*}
After subtracting \(\lambda_3\), all remaining weights are zero.
\end{itemize}
So the original weights decompose as
\begin{align*}
    \theta_1=\lambda_1+\lambda_3,\qquad
    \theta_2=\lambda_2,\qquad
    \theta_3=\lambda_3,\qquad
    \theta_4=\lambda_1,\qquad
    \theta_5=\lambda_2+\lambda_3.
\end{align*}
Each layer \(\Pi_t\) gives an SPR mechanism \(\widehat Q^{(t)}\), relabeled onto the common output alphabet \([5]\):
\begin{align*}
    \widehat Q^{(1)}=
    \begin{bmatrix}
    \frac{e}{e+1} & 0 & 0 & \frac1{e+1} & 0\\[4pt]
    \frac1{e+1} & 0 & 0 & \frac{e}{e+1} & 0\\
    \frac1{e+1} & 0 & 0 & \frac{e}{e+1} & 0\\
    \frac1{e+1} & 0 & 0 & \frac{e}{e+1} & 0\\
    \frac1{e+1} & 0 & 0 & \frac{e}{e+1} & 0
    \end{bmatrix},
\end{align*}
\begin{align*}
    \widehat Q^{(2)}=
    \begin{bmatrix}
    0 & \frac{e}{e+1} & 0 & 0 & \frac1{e+1}\\[4pt]
    0 & \frac{e}{e+1} & 0 & 0 & \frac1{e+1}\\
    0 & \frac1{e+1} & 0 & 0 & \frac{e}{e+1}\\
    0 & \frac1{e+1} & 0 & 0 & \frac{e}{e+1}\\
    0 & \frac1{e+1} & 0 & 0 & \frac{e}{e+1}
    \end{bmatrix},
\end{align*}
\begin{align*}
    \widehat Q^{(3)}=
    \begin{bmatrix}
    \frac{e}{e+2} & 0 & \frac1{e+2} & 0 & \frac1{e+2}\\[4pt]
    \frac1{e+2} & 0 & \frac{e}{e+2} & 0 & \frac1{e+2}\\
    \frac1{e+2} & 0 & \frac1{e+2} & 0 & \frac{e}{e+2}\\
    \frac1{e+2} & 0 & \frac1{e+2} & 0 & \frac{e}{e+2}\\
    \frac1{e+2} & 0 & \frac1{e+2} & 0 & \frac{e}{e+2}
    \end{bmatrix}.
\end{align*}
With
\begin{align*}
    \beta_1=\lambda_1(e+1)=\frac{7(e+1)}{D},\qquad
    \beta_2=\lambda_2(e+1)=\frac{4(e+1)}{D},\qquad
    \beta_3=\lambda_3(e+2)=\frac{3(e+2)}{D},
\end{align*}
we have
\begin{align*}
    \beta_1+\beta_2+\beta_3
    =\frac{7(e+1)+4(e+1)+3(e+2)}{D}
    =\frac{14e+17}{D}=1,
\end{align*}
and
\begin{align*}
    Q=\beta_1\widehat Q^{(1)}+\beta_2\widehat Q^{(2)}+\beta_3\widehat Q^{(3)}.
\end{align*}
\end{example}

%% file: appendix/C_coarsening_spr.tex
\section{Nontrivial coarsenings of SPR mechanisms}
\label{sec:proof-no-new-vertices-under-coarsening}

We show that a nontrivial deterministic merging of the outputs of an SPR mechanism never creates a new extreme point. The reason is that, after merging, one copy of each coarse output still carries the usual randomized-response signal, while the remaining merged copies contribute only hypothesis-independent noise.

\begin{proof}[Proof of Proposition~\ref{prop:no-new-vertices-under-coarsening}]
Let $\pi=(B_1,\dots,B_s)$ be an LR-contiguous partition of $[k]$, and let $Q^\pi$ be the corresponding SPR channel. Equivalently, $Q^\pi$ applies $s$-ary randomized response to the partition label:
\begin{align*}
    Q^\pi(j\mid x)=
    \begin{cases}
    \dfrac{e^\eps}{e^\eps+s-1}, & x\in B_j,\\[6pt]
    \dfrac{1}{e^\eps+s-1}, & x\notin B_j.
    \end{cases}
\end{align*}

Let $\tau:[s]\to[\ell]$ be a deterministic merging map, and write $\tau\circ Q^\pi$ for the induced channel
\begin{align*}
    (\tau\circ Q^\pi)(y\mid x):=\sum_{j:\tau(j)=y} Q^\pi(j\mid x)
    \qquad (x\in[k],\ y\in[\ell]).
\end{align*}
Write
\begin{align*}
    A:=\tau([s])\subseteq [\ell]
\end{align*}
for the coarse output labels that remain active, with $t:=|A|$. Assume $1<t<s$, so at least one pair of fine outputs is merged, but the channel is not collapsed to a single active output. For each $y\in A$, define
\begin{align*}
    n_y:=|\tau^{-1}(y)|,
    \qquad
    C_y:=\bigcup_{j:\tau(j)=y} B_j.
\end{align*}
Thus $n_y$ is the number of fine labels merged into $y$, and the sets $\{C_y:y\in A\}$ form the coarser partition of $[k]$ induced by $\tau$.

We now define two $\ell$-output channels. The first is randomized response on the coarser partition $\{C_y:y\in A\}$, and the second is an input-independent channel capturing the remaining mass.
\begin{itemize}
    \item Let $R$ be randomized response on the coarser partition:
    \begin{align*}
    R(y\mid x)=
    \begin{cases}
        \dfrac{1+(e^\eps-1)\mathbf 1\{x\in C_y\}}{e^\eps+t-1}, & y\in A,\\[8pt]
        0, & y\notin A.
        \end{cases}
    \end{align*}

    \item Let $U$ be the input-independent channel
    \begin{align*}
    U(y\mid x)=
    \begin{cases}
        \dfrac{n_y-1}{s-t}, & y\in A,\\[8pt]
        0, & y\notin A,
        \end{cases}
        \qquad x\in [k].
    \end{align*}
    This is well defined because $\sum_{y\in A}(n_y-1)=\sum_{y\in A}n_y-t=s-t$.
\end{itemize}

Set $\displaystyle \alpha:=\frac{e^\eps+t-1}{e^\eps+s-1}\in(0,1)$.
We claim that
\begin{align*}
    \tau\circ Q^\pi=\alpha R+(1-\alpha)U.
\end{align*}
Fix $x\in[k]$ and $y\in[\ell]$. If $y\notin A$, then both sides are zero. Now suppose $y\in A$.

\begin{itemize}
    \item If $x\in C_y$, then among the $n_y$ fine outputs merged into $y$, exactly one corresponds to the block containing $x$, and therefore has probability $\frac{e^\eps}{e^\eps+s-1}$; the other $n_y-1$ have probability $\frac{1}{e^\eps+s-1}$. Hence
    \begin{align*}
        (\tau\circ Q^\pi)(y\mid x)=\frac{e^\eps+n_y-1}{e^\eps+s-1}
        =\alpha\frac{e^\eps}{e^\eps+t-1}+(1-\alpha)\frac{n_y-1}{s-t}.
    \end{align*}
    \item If $x\notin C_y$, then none of the $n_y$ fine outputs merged into $y$ corresponds to the block containing $x$, so all $n_y$ of them have probability $\frac{1}{e^\eps+s-1}$. Therefore
    \begin{align*}
        (\tau\circ Q^\pi)(y\mid x)=\frac{n_y}{e^\eps+s-1}
        =\alpha\frac{1}{e^\eps+t-1}+(1-\alpha)\frac{n_y-1}{s-t}.
    \end{align*}
\end{itemize}
This proves the claim.

Let
\begin{align*}
    M_0:=((\tau\circ Q^\pi)^\top p),
    \qquad
    M_1:=((\tau\circ Q^\pi)^\top q),
\end{align*}
and similarly let
\begin{align*}
    M_0^R:=R^\top p,\quad M_1^R:=R^\top q,
    \qquad
    M_0^U:=U^\top p,\quad M_1^U:=U^\top q.
\end{align*}
Since $U$ is input-independent, we have $M_0^U=M_1^U$. Write this common distribution as $c:=M_0^U=M_1^U$, or equivalently,
\begin{align*}
    c(y)=M_0^U(y)=M_1^U(y)=
    \begin{cases}
    \dfrac{n_y-1}{s-t}, & y\in A,\\[8pt]
    0, & y\notin A.
    \end{cases}
\end{align*}
Taking induced marginals under $p$ and $q$ in the identity
\begin{align*}
    \tau\circ Q^\pi=\alpha R+(1-\alpha)U
\end{align*}
gives
\begin{align*}
    (M_0,M_1)=\alpha(M_0^R,M_1^R)+(1-\alpha)(c,c).
\end{align*}
Both terms on the right lie in $\mathcal R_\ell(p,q)$: the pair $(M_0^R,M_1^R)$ comes from the $\eps$-LDP channel $R$, and $(c,c)$ comes from the input-independent channel $U$.
\begin{itemize}
    \item If $(M_0^R,M_1^R)\neq(c,c)$, this is already a nontrivial convex decomposition, so $(M_0,M_1)$ is not extreme.

    \item If $(M_0^R,M_1^R)=(c,c)$, then $c=M_0^R=M_1^R$, and $(M_0,M_1)=\alpha(c,c)+(1-\alpha)(c,c)=(c,c)$, so it suffices to show that $(c,c)$ is not an extreme point of $\mathcal R_\ell(p,q)$. For every $y\in A$,
\begin{align*}
    R(y\mid x)\ge \frac{1}{e^\eps+t-1}\qquad \text{for all }x\in[k],
\end{align*}
so
\begin{align*}
    c(y)=M_0^R(y)\ge \frac{1}{e^\eps+t-1}>0.
\end{align*}
Since $t=|A|>1$, choose distinct $y_1,y_2\in A$ and
\begin{align*}
    0<\eta<\min\{c(y_1),c(y_2)\}.
\end{align*}
Obtain two new probability distributions by moving mass $\eta$ from $y_2$ to $y_1$ and in the opposite direction:
\begin{align*}
    c^+(y)=
    \begin{cases}
    c(y_1)+\eta, & y=y_1,\\
    c(y_2)-\eta, & y=y_2,\\
    c(y), & \text{otherwise},
    \end{cases}
    \qquad
    c^-(y)=
    \begin{cases}
    c(y_1)-\eta, & y=y_1,\\
    c(y_2)+\eta, & y=y_2,\\
    c(y), & \text{otherwise}.
    \end{cases}
\end{align*}
Then $c^+\neq c^-$ and
\begin{align*}
    c=\frac12(c^++c^-).
\end{align*}
Since $c^\pm$ are probability distributions, the input-independent channels
\begin{align*}
    U^\pm(y\mid x):=c^\pm(y)
\end{align*}
are feasible, so $(c^+,c^+),(c^-,c^-)\in \mathcal R_\ell(p,q)$. Therefore
\begin{align*}
    (c,c)=\frac12(c^+,c^+)+\frac12(c^-,c^-)
\end{align*}
is not an extreme point of $\mathcal R_\ell(p,q)$.

\end{itemize}

In either case, the image point induced by $\tau\circ Q^\pi$ is not an extreme point of $\mathcal R_\ell(p,q)$.
\end{proof}

%% file: appendix/D_egamma_proofs.tex
\section{Optimality of the generalized binary mechanism for \texorpdfstring{$E_\gamma$}{E\_gamma}}
\label{sec:proof-optimal-egamma}

\begin{proof}[Proof of Theorem~\ref{thm:optimal-egamma}]
Let
\begin{align*}
    S_\gamma := \{ x \in \X : P_0(x) \ge \gamma P_1(x) \}.
\end{align*}
If $S_\gamma = \emptyset$ or $S_\gamma = \X$, then $E_\gamma(P_0\|P_1) = 0$, and the theorem is immediate: in the second case, summing $P_0(x) \ge \gamma P_1(x)$ over $x$ gives $1 \ge \gamma$, which together with $\gamma \ge 1$ forces $\gamma = 1$ and $P_0 = P_1$, hence $M_0 = M_1$. So assume both $S_\gamma$ and $S_\gamma^c$ are nonempty.

\citet{Zamanlooy2024E_gammaMixing} (Theorem~3) show the upper bound in Theorem~\ref{thm:optimal-egamma}. We show that the generalized binary mechanism attains this bound. Under $Q_\gamma$,
\begin{align*}
    M_0(0) - \gamma M_1(0)
    = \sum_x \bigl( P_0(x) - \gamma P_1(x) \bigr) Q_\gamma(0 \mid x).
\end{align*}

By the Neyman-Pearson characterization of $E_\gamma$,
\begin{equation}
    \label{eq:egamma-np}
    E_\gamma(P_0\|P_1) = P_0(S_\gamma) - \gamma P_1(S_\gamma),
    \qquad
    P_0(S_\gamma^c) - \gamma P_1(S_\gamma^c) = (1 - \gamma) - E_\gamma(P_0\|P_1) .
\end{equation}

Splitting the sum over $S_\gamma$ and $S_\gamma^c$ and using \eqref{eq:egamma-np}, we get
\begin{align*}
    M_0(0) - \gamma M_1(0)
    = \frac{e^\eps}{1 + e^\eps} \bigl( P_0(S_\gamma) - \gamma P_1(S_\gamma) \bigr)
    + \frac{1}{1 + e^\eps} \bigl( P_0(S_\gamma^c) - \gamma P_1(S_\gamma^c) \bigr)
\end{align*}
\begin{align*}
    = \frac{e^\eps E_\gamma(P_0\|P_1) + (1 - \gamma - E_\gamma(P_0\|P_1))}{1 + e^\eps}
    = \frac{e^\eps - 1}{e^\eps + 1} E_\gamma(P_0 \| P_1) + \frac{1 - \gamma}{e^\eps + 1}.
\end{align*}
Since $E_\gamma(M_0 \| M_1)$ is the supremum over all events, it is at least the value of the event $\{0\}$ and also at least $0$ (by taking the empty event). Hence
\begin{align*}
    E_\gamma(M_0 \| M_1)
    \ge \max\{ 0, \, M_0(0) - \gamma M_1(0) \}
    = \left( \frac{e^\eps - 1}{e^\eps + 1} E_\gamma(P_0 \| P_1) + \frac{1 - \gamma}{e^\eps + 1} \right)_+.
\end{align*}
Combined with the upper bound, this proves equality and therefore optimality.
\end{proof}

We show that one can recover the bound from our results as well:

\begin{corollary}[Recovering the $E_\gamma$ upper bound from SPR optimality]
SPR optimality implies
\begin{align*}
    E_\gamma(M_0 \| M_1)
    \le
    \left(
        \frac{e^\eps - 1}{e^\eps + 1} E_\gamma(P_0 \| P_1)
        + \frac{1 - \gamma}{e^\eps + 1}
    \right)_+ .
\end{align*}
\end{corollary}

\begin{proof}
Since $E_\gamma$ is an $f$-divergence, it is enough to optimize over SPR mechanisms. Consider an arbitrary SPR mechanism induced by an LR-contiguous partition $\pi=(B_1,\dots,B_s)$ into $s$ blocks.

If $s=1$, then the SPR mechanism has a single output and is input-independent, so $M_0^\pi=M_1^\pi$. Since $\gamma\ge 1$, this implies $E_\gamma(M_0^\pi\|M_1^\pi)=0$, which is at most the right-hand side of the claimed bound because the latter is nonnegative by definition of $(\cdot)_+$. Hence it remains to consider $s\ge 2$.

The induced output probabilities at the output corresponding to block $B$ are
\begin{align*}
    M_0^\pi(B)
    =
    \frac{1+(e^\eps-1)P_0(B)}{e^\eps+s-1},
    \qquad
    M_1^\pi(B)
    =
    \frac{1+(e^\eps-1)P_1(B)}{e^\eps+s-1}.
\end{align*}
Therefore the contribution of this output to $E_\gamma(M_0^\pi\|M_1^\pi)$ is
\begin{align*}
    \bigl(M_0^\pi(B)-\gamma M_1^\pi(B)\bigr)_+
    &=
    \frac{\bigl(1-\gamma+(e^\eps-1)\bigl(P_0(B)-\gamma P_1(B)\bigr)\bigr)_+}{e^\eps+s-1} \\
    &=
    \frac{e^\eps-1}{e^\eps+s-1}\Bigl(P_0(B)-\gamma P_1(B)-\dfrac{\gamma-1}{e^\eps-1}\Bigr)_+,
\end{align*}
where the last equality factors out $e^\eps-1>0$. Summing over blocks,
\begin{align}
    \label{eq:egamma-spr-partition-value}
    E_\gamma(M_0^\pi\|M_1^\pi)
    =
    \frac{e^\eps-1}{e^\eps+s-1}
    \sum_{B\in\pi} \Bigl(P_0(B)-\gamma P_1(B)-\dfrac{\gamma-1}{e^\eps-1}\Bigr)_+.
\end{align}

We upper-bound the partition score: Let $\mathrm{Pos}(\pi)$ denote the blocks of $\pi$ that contribute positively to the sum, i.e., $\mathrm{Pos}(\pi):=\{B\in\pi:P_0(B)-\gamma P_1(B)>\frac{\gamma-1}{e^\eps-1}\}$.
If $\mathrm{Pos}(\pi)=\emptyset$, then the sum in \eqref{eq:egamma-spr-partition-value} is zero.

Otherwise, since $|\mathrm{Pos}(\pi)|\ge 1$ and $\dfrac{\gamma-1}{e^\eps-1}\ge 0$,
\begin{align*}
    \sum_{B\in\pi} \Bigl(P_0(B)-\gamma P_1(B)-\dfrac{\gamma-1}{e^\eps-1}\Bigr)_+
    &=
    \sum_{B\in\mathrm{Pos}(\pi)} \bigl(P_0(B)-\gamma P_1(B)\bigr) - |\mathrm{Pos}(\pi)|\,\dfrac{\gamma-1}{e^\eps-1} \\
    &\le
    \sum_{B\in\mathrm{Pos}(\pi)} \bigl(P_0(B)-\gamma P_1(B)\bigr) - \dfrac{\gamma-1}{e^\eps-1} \\
    &\le
    \sum_{i=1}^k \bigl(P_0(i)-\gamma P_1(i)\bigr)_+ - \dfrac{\gamma-1}{e^\eps-1} \\
    &=
    E_\gamma(P_0\|P_1) - \dfrac{\gamma-1}{e^\eps-1} .
\end{align*}
Thus, in all cases,
\begin{align}
    \label{eq:block-positive-bound}
    \sum_{B\in\pi} \Bigl(P_0(B)-\gamma P_1(B)-\dfrac{\gamma-1}{e^\eps-1}\Bigr)_+
    \le
    \Bigl(E_\gamma(P_0\|P_1)-\dfrac{\gamma-1}{e^\eps-1}\Bigr)_+ .
\end{align}

Combining \eqref{eq:egamma-spr-partition-value} and \eqref{eq:block-positive-bound} gives
\begin{align*}
    E_\gamma(M_0^\pi\|M_1^\pi)
    &\le
    \frac{e^\eps-1}{e^\eps+s-1}\Bigl(E_\gamma(P_0\|P_1)-\dfrac{\gamma-1}{e^\eps-1}\Bigr)_+ \\
    &\le
    \frac{e^\eps-1}{e^\eps+1}\Bigl(E_\gamma(P_0\|P_1)-\dfrac{\gamma-1}{e^\eps-1}\Bigr)_+ \\
    &=
    \left(
        \frac{e^\eps-1}{e^\eps+1}E_\gamma(P_0\|P_1)
        + \frac{1-\gamma}{e^\eps+1}
    \right)_+ .
\end{align*}
By SPR optimality for $f$-divergence objectives, the same bound holds for every $\eps$-LDP mechanism.
\end{proof}

Equation~\eqref{eq:egamma-spr-partition-value} conveys the intuition behind the optimal mechanism being binary: the factor $\frac{e^\eps-1}{e^\eps+s-1}$ decreases with the number of outputs, and the $\frac{\gamma-1}{e^\eps-1}$ penalty scales with the number of outputs.

%% file: appendix/E_numerical_experiments.tex
\section{Numerical experiments}
\label{sec:experiments}

\subsection{Monte-Carlo verification of the joint-range polytope}
\label{sec:joint-range-monte-carlo}

Figure~\ref{fig:joint-range-three-eps} empirically corroborates Theorem~\ref{thm:extreme-points-spr} at three privacy levels for the same $(p,q)$ used in Figure~\ref{fig:joint-range-k5-eps2}. For each $\varepsilon\in\{0.4,2,5\}$ we sample random $\varepsilon$-LDP channels $Q\in\mathcal Q_{\varepsilon,2}$, plot their induced image points $(Q^\top p,Q^\top q)$, and overlay the predicted polytope: the convex hull of the $2(k-1)+2=10$ SPR-induced candidate points from Theorem~\ref{thm:extreme-points-spr}. As $\varepsilon$ grows the feasible set of channels expands, the cloud of Monte-Carlo samples spreads accordingly, and additional SPR vertices become extreme; in every panel the cloud is contained in the predicted hull.

\paragraph{Sampling procedure.}
Each panel uses $N=2{,}000{,}000$ candidate channels drawn as follows. For $\varepsilon=0.4$, each of the $k$ rows is sampled uniformly on $[0,1]^2$ and then renormalized so that the two row entries sum to $1$; for $\varepsilon\in\{2,5\}$, the $k$ rows are instead drawn from a $\mathrm{Beta}(0.3,0.3)$ corner-seeking distribution and renormalized analogously, which concentrates proposals near the boundary of the feasible set where the LDP constraint is most likely to bind. A candidate is then kept only if every column ratio lies in $[e^{-\varepsilon},e^\varepsilon]$, i.e.\ the channel is $\varepsilon$-LDP. The corresponding rejection-sampling acceptance rates were $0.147\%$, $1.675\%$, and $26.336\%$ for $\varepsilon=0.4,2,5$ respectively. Proposals are generated with NumPy's \texttt{default\_rng(20260503)}, so the figure is bit-exactly reproducible.

\begin{figure}[tbp]
  \centering
  \includegraphics[width=\linewidth]{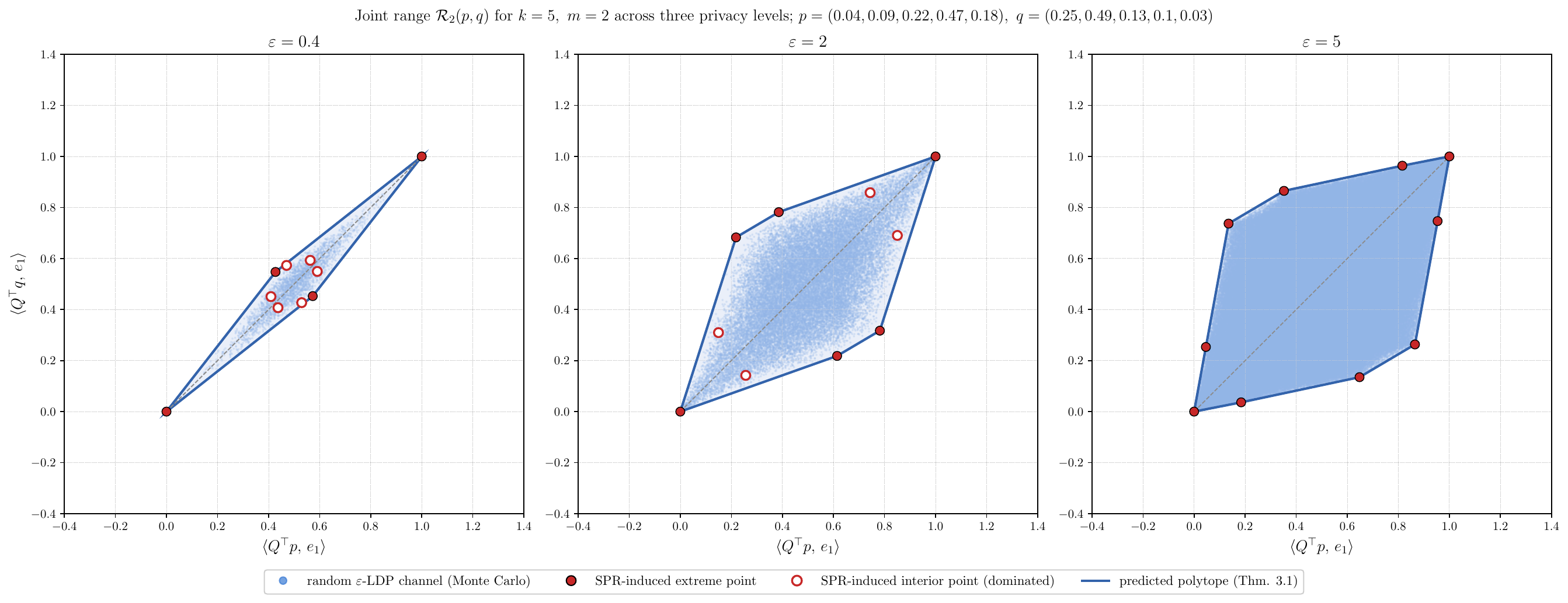}
  \caption{Monte-Carlo verification of $\mathcal R_2(p,q)$ at three privacy levels for $k=5$, $\ell=2$, $p=(0.04,0.09,0.22,0.47,0.18)$, $q=(0.25,0.49,0.13,0.10,0.03)$. Each panel shows random $\varepsilon$-LDP channels (gray cloud), the $2(k-1)+2=10$ SPR-induced candidate points (red: filled if extreme, hollow if dominated), and the predicted polytope (Theorem~\ref{thm:extreme-points-spr}, in blue). The hull contains every sampled channel at each $\varepsilon$. At $\varepsilon=0.4$ only $4$ of $10$ SPR candidates are extreme; at $\varepsilon=2$ that grows to $6$; at $\varepsilon=5$ all $10$ are extreme. Again, because $\ell=2$, the plotted coordinates give the full induced pair $(Q^\top p,Q^\top q)$, so points marked as non-extreme are not artifacts of a lower-dimensional projection.}
  \label{fig:joint-range-three-eps}
\end{figure}

\subsection{KL comparison: setup, bootstrap CIs, and runtimes}
\label{sec:experiments-kl}

We numerically compare the SPR dynamic program (Algorithm~\ref{alg:optimal-spr-dp}) against the closed-form binary mechanism, $k$-ary randomized response, and (where tractable) the Kairouz--Oh--Viswanath linear program~\cite{Kairouz2016ExtremalMechanisms}, on the KL utility
\begin{align*}
    \frac{D_{\mathrm{kl}}(M_0\|M_1)}{D_{\mathrm{kl}}(P_0\|P_1)},
    \qquad
    M_\nu=Q^{\!\top}P_\nu,
\end{align*}
which we maximize over $\eps$-LDP channels $Q$.

\paragraph{Setup.}
For each value of the alphabet size $k$, we draw $T=100$ pairs
$(P_0^{(t)},P_1^{(t)})\sim\mathrm{Dirichlet}(\mathbf 1_k)\otimes\mathrm{Dirichlet}(\mathbf 1_k)$ once at the start of the experiment (NumPy seed $0$), uniform on the probability simplex. The same $100$ pairs are reused for every mechanism and every $\eps$, so the comparison is paired across mechanisms. We sweep a uniform privacy grid $\eps\in\{0,0.1,\dots,9.9,10.0\}$ ($101$ values), and at each $(\eps,t)$ evaluate every mechanism on the same $(P_0^{(t)},P_1^{(t)})$. Curves report the mean of the normalized KL utility over the $T=100$ trials, skipping the (measure-zero) trials with $D_{\mathrm{kl}}(P_0^{(t)}\|P_1^{(t)})=0$.

\paragraph{Confidence bands.}
For each mechanism we additionally report a $95\%$ percentile-bootstrap confidence interval for the mean curve (Figure~\ref{fig:kl-eps-two-panel-ci}), computed by resampling the $T=100$ trials with replacement $5000$ times and taking the $2.5$th and $97.5$th percentiles of the resampled mean at each $\eps$. The bootstrap is non-parametric and assumes only that the trials are exchangeable; no Gaussianity is assumed, and the band is asymmetric in general. The variability captured is over the random Dirichlet draws of $(P_0,P_1)$ only. The same bootstrap indices are used for every mechanism so that the paired structure across mechanisms is preserved at the band level. The bands are intentionally light because the $k=6$ optimality claim is not statistical: it is the load-bearing fact that the SPR dynamic program agrees with the KOV LP on every trial. Empirically, $\max_{t,\eps}|V_{\mathrm{SPR}}^{(t)}(\eps)-V_{\mathrm{KOV}}^{(t)}(\eps)|\le 6.1\times 10^{-14}$ across the full $T\times 101$ grid, consistent with floating-point round-off only.

\paragraph{Mechanisms compared.}
The KOV LP solves the $2^k$-variable staircase linear program with HiGHS via \texttt{scipy.optimize.linprog}; this is the exact $\eps$-LDP optimum but is only tractable for small $k$. The SPR DP runs in $O(\ell\,k^2)$ time after sorting the alphabet by likelihood ratio, with output budget $\ell=k$; for $k=100$ we use a vectorized cumulative-max along columns of $A[r,i]=F[s-1,r]+\mu[r+1:i]$, which keeps each $s$-step at one NumPy call. The binary mechanism uses the high set $H_0=\{x:P_0(x)\ge P_1(x)\}$; randomized response is the standard $k$-ary channel with diagonal probability $e^\eps/(k-1+e^\eps)$.

\paragraph{Results (Figure~\ref{fig:kl-eps-two-panel}).}
In Figure~\ref{fig:kl-eps-two-panel}\subref{fig:kl-eps-k6} ($k=6$) the LP is tractable and we use it as a baseline. The SPR DP curve coincides with the LP curve to floating-point precision: the maximum absolute gap across the $10{,}100$ evaluations $(t,\eps)$ is at most $6.1\times 10^{-14}$, consistent with round-off only. This is the empirical version of Theorem~\ref{thm:extreme-points-spr}: SPR mechanisms attain the same optimum as the LP. The binary mechanism plateaus around $0.72$ for $\eps\gtrsim 5$, while $k$-ary randomized response is uninformative for small $\eps$ (below $25\%$ of the SPR optimum at $\eps=0.5$) and only reaches within $5\%$ of the SPR curve by $\eps\approx 4$, with full parity around $\eps\approx 6$.

In Figure~\ref{fig:kl-eps-two-panel}\subref{fig:kl-eps-k100} ($k=100$) the LP has $2^{100}\approx 1.27\times 10^{30}$ variables and is not solvable; the SPR DP, in contrast, finishes the entire $101\times 100$ sweep in under $20$ seconds on a laptop (median $17.8$\,s over three runs on a MacBook Pro with an Apple M3 Pro chip and $36$\,GB RAM, isolated from the other mechanisms). By Theorem~\ref{thm:extreme-points-spr}, the SPR curve is the value the LP \emph{would} achieve if it could be run. Two regime changes are visible relative to $k=6$: the binary plateau drops from $\approx 0.72$ to $\approx 0.56$, since a single likelihood-ratio threshold discards more information at larger $k$; and randomized response is essentially uninformative for small $\eps$, only catching up to the optimum once $e^\eps\gtrsim k-1$, i.e.\ for $\eps$ above roughly $\ln(k-1)\approx 4.6$.

\begin{figure}[tbp]
  \centering
  \begin{subfigure}{0.49\linewidth}
    \centering
    \includegraphics[width=\linewidth]{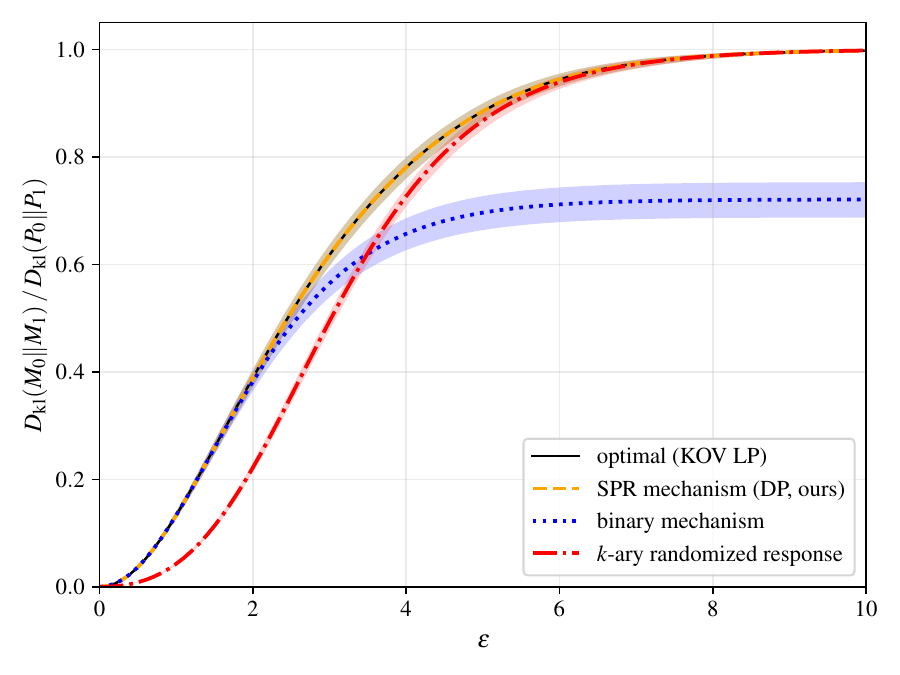}
    \caption{$k=6$.}
    \label{fig:kl-eps-k6-ci}
  \end{subfigure}\hfill
  \begin{subfigure}{0.49\linewidth}
    \centering
    \includegraphics[width=\linewidth]{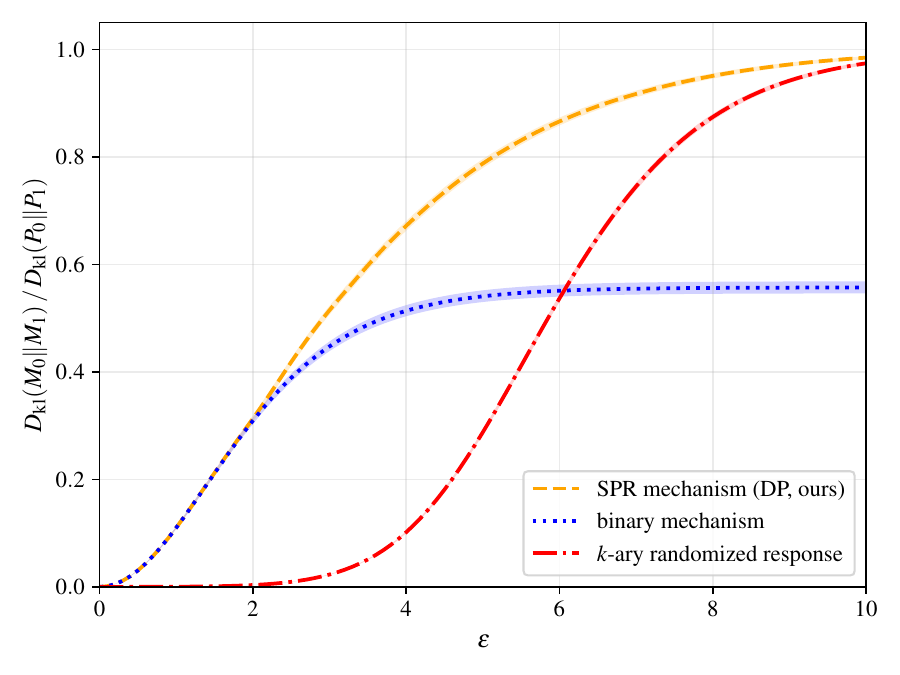}
    \caption{$k=100$.}
    \label{fig:kl-eps-k100-ci}
  \end{subfigure}
  \caption{Same curves as Figure~\ref{fig:kl-eps-two-panel}, with shaded $95\%$ percentile-bootstrap confidence intervals (5000 resamples) for the mean across the $T=100$ paired Dirichlet draws. The same bootstrap indices are used for every mechanism, so the paired structure of the experiment is preserved at the band level. At $k=6$ the SPR and KOV bands coincide to within bootstrap noise; the substantive SPR-vs-KOV comparison is the per-trial agreement $\max_{t,\eps}|V_{\mathrm{SPR}}^{(t)}-V_{\mathrm{KOV}}^{(t)}|\le 6.1\times 10^{-14}$ reported in the text, not the bands.}
  \label{fig:kl-eps-two-panel-ci}
\end{figure}